\documentclass{CSML}

\def\dOi{12(4:9)2016}
\lmcsheading%
{\dOi}
{1--27}
{}
{}
{May\phantom{.0}5, 2016}
{Dec.~28, 2016}
{}

\ACMCCS{[{\bf Theory of computation}]: Logic---Proof theory}

\usepackage{hyperref}
\theoremstyle{plain}

\usepackage[lutzsyntax,noxy]{virginialake}
\usepackage{amsmath, amsthm, amssymb}
\usepackage{hyperref}
\usepackage{color} 
\usepackage{tikz}
\usepackage{xstring}
\usepackage[utf8]{inputenc}

\newcommand{\fff}{\bot}
\newcommand{\ttt}{\top}

\newcommand{\ac  }{\mathsf{ac}}

\newcommand{\acd }{{\ac{\downarrow}}}

\newcommand{\acu }{{\ac{\uparrow}}}
\newcommand{\swi }{\mathsf{s}}
\newcommand{\med }{\mathsf{m}}

\newcommand  {\gw  }{\mathsf w}
\newcommand  {\gwd }{{\gw{\downarrow}}}
\newcommand  {\gwu }{{\gw{\uparrow}}}
\newcommand  {\gc  }{\mathsf c}
\newcommand  {\gcdown }{{\gc{\downarrow}}}
\newcommand  {\gcu }{{\gc{\uparrow}}}

\newcommand{\SKS}{\mathsf{SKS}}

\newcommand{\set}[1]{\{#1\}}

\newcommand{\ACU}{ACU}

\newcommand{\NP}{\mathbf{NP}}
\newcommand{\coNP}{\mathbf{coNP}}

\newcommand{\MIN}{\mathit{MIN}}
\newcommand{\MAX}{\mathit{MAX}}
\newcommand{\AC}{\mathit{AC}}

\newcommand{\fnimplies}[2]{{#1}\to{#2}}

\newcommand{\cmplmt}[1]{\overline{#1}}

\newcommand{\supersub}[3]{ \underset{#2}{\overset{#1}{#3}} }
\newcommand{\annotatedarrow}[2]{ \supersub{\!#1}{\!#2\;}{\rightarrow}}
\newcommand{\longannotatedarrow}[2]{ \supersub{#1}{#2}{\longrightarrow} }

\newcommand{\Lin}{\mathsf{L}}

\newcommand{\Var}{\mathit{Var}}
\newcommand{\Ter}{\mathit{Ter}}

\newcommand{\pfour}{P_4}

\tikzstyle{vertex}=[circle,minimum size=10pt,inner sep=0pt,text height=1.5ex,text depth=.25ex]
\tikzstyle{g} = [draw,thick,dotted,-,green]
\tikzstyle{r} = [draw,thick,-,red]
\tikzstyle{b} = [draw,thick,-,black]
\tikzstyle{u} = [draw,thin,-,black,opacity=0]

\renewcommand{\gcd}{\gcdown }

\newcommand{\TwoGraph}[2]{ 
{\begin{tikzpicture}[baseline=(v1.base)]
		\node[vertex] (v1) at (0, 0) {$\StrBetween[1,2]{,#1,}{,}{,}$};
		\node[vertex]  (v2) at (1, 0) {$\StrBetween[2,3]{,#1,}{,}{,}$};
		\draw[#2]  (v1) -- (v2) ;
\end{tikzpicture} }}

\newcommand{\TwoGraphLabel}[3][]{ 
{\begin{tikzpicture}[baseline=(v1.base)]
		\node[vertex] (v1) at (0, 0) {$\StrBetween[1,2]{,#2,}{,}{,}$};
		\node[vertex]  (v2) at (1, 0) {$\StrBetween[2,3]{,#2,}{,}{,}$};
		\draw[#3]  (v1) -- (v2) node [midway,above=-1pt] {$\scriptstyle #1$};
\end{tikzpicture} }}

\newcommand{\redge}[3][]{\TwoGraphLabel[\smash{#1}]{#2,#3}r}
\newcommand{\gedge}[3][]{\TwoGraphLabel[\smash{#1}]{#2,#3}g}

\newcommand{\FourGraph}[7]{ 
\raisebox{-0.4\height}{\begin{tikzpicture}
		\node[vertex]  (v1) at (0, 1) {$\StrBetween[1,2]{,#1,}{,}{,}$};
		\node[vertex]  (v2) at (1, 1) {$\StrBetween[2,3]{,#1,}{,}{,}$};
		\node[vertex] (v3) at (0, 0) {$\StrBetween[3,4]{,#1,}{,}{,}$};
		\node[vertex]  (v4) at (1, 0) {$\StrBetween[4,5]{,#1,}{,}{,}$};
		\draw[#2]  (v1) -- (v2) ;
		\draw[#3]  (v1) -- (v3);
		\draw[#4] (v1) -- (v4);
		\draw[#5]  (v2) -- (v3);
		\draw[#6] (v2) -- (v4);
		\draw[#7]  (v3) -- (v4);
\end{tikzpicture}} }

\newcommand{\FiveGraphRCont}[2]{
\raisebox{-0.45\height}{\begin{tikzpicture}
		\node[vertex] (v1) at (0,1) {$\StrBetween[1,2]{,\fir,}{,}{,}$};
		\node[vertex] (v2) at (1,1) {$\StrBetween[2,3]{,\fir,}{,}{,}$};
		\node[vertex] (v3) at (0,0) {$\StrBetween[3,4]{,\fir,}{,}{,}$};
		\node[vertex] (v4) at (1,0) {$\StrBetween[4,5]{,\fir,}{,}{,}$};
		\node[vertex] (v5) at (1.25,-.8) {$\StrBetween[5,6]{,\fir,}{,}{,}$};
		\draw[\sec] (v1) edge (v2) ;
		\draw[\thi] (v1) -- (v3);
		\draw[\fou] (v1) -- (v4);
		\draw[\fif] (v2) -- (v3);
		\draw[\six] (v2) -- (v4);
		\draw[\sev] (v3) -- (v4);
		\draw[\eig, bend angle=5, bend right] (v3) edge (v5) ;
		\draw[\nin, bend angle=10, bend right] (v1) edge (v5) ;
		\draw[#1] (v4) -- (v5) ;
		\draw[#2, bend angle=20, bend left] (v2) edge (v5) ;
\end{tikzpicture} }}
\newcommand{\FiveGraphR}[9]{
	\def\fir{#1}
	\def\sec{#2}
	\def\thi{#3}
	\def\fou{#4}
	\def\fif{#5}
	\def\six{#6}
	\def\sev{#7}
	\def\eig{#8}
	\def\nin{#9}
	\FiveGraphRCont
}

\def\SKS{\mathsf{SKS}}

\def\set#1{\{#1\}}

\def\grammareq {\mathrel{\raise.4pt\hbox{::}{=}}}%

\newcommand{\dotto}[1][]{\mathrel{\!\xy\ar@{.>}^-{#1}(5,0)\endxy\!}}
\newcommand{\solto}[1][]{\mathrel{\!\xy\ar@{->}^-{#1}(5,0)\endxy\!}}
\newcommand{\longsolto}[1][]{\mathrel{\!\xy\ar@{->}^-{#1}(11,0)\endxy\!}}
\newcommand{\longdotto}[1][]{\mathrel{\!\xy\ar@{.>}^-{#1}(11,0)\endxy\!}}
\newcommand{\xldotto}[2][]{\mathrel{\!\xy\ar@{.>}^-{#1}(#2,0)\endxy\!}}

\newcommand{\assoc}[1][]{\alpha_{#1}}

\def\cand{\wedge}
\def\cor{\vee}

\newbox\cutbox
\newdimen\cutwd
\newdimen\cutht
\newdimen\cutdp
\def\ccut{%
  \setbox\cutbox\hbox{$\lozenge$}
  \cutwd=\wd\cutbox
  \cutht=\ht\cutbox
  \cutdp=\dp\cutbox
  \setbox\cutbox\hbox to\cutwd{\hss\vrule width.3pt height\cutht depth\cutdp\hss}
  \mathbin{\lozenge\hskip-\cutwd\copy\cutbox}}
\def\scriptcut{%
  \setbox\cutbox\hbox{$\scriptstyle\lozenge$}
  \cutwd=\wd\cutbox
  \cutht=\ht\cutbox
  \cutdp=\dp\cutbox
  \setbox\cutbox\hbox to\cutwd{\hss\vrule width.3pt height\cutht depth\cutdp\hss}
  \mathord{\lozenge\hskip-\cutwd\copy\cutbox}}

\def\vccut{%
  \setbox\cutbox\hbox{$\lozenge$}
  \cutwd=\wd\cutbox
  \cutht=\ht\cutbox
  \cutdp=\dp\cutbox
  \setbox\cutbox\hbox to\cutwd{\hss\hskip.3pt\vrule width.3pt height\cutht depth\cutdp\hss}
  \mathbin{\lozenge\hskip-\cutwd\copy\cutbox}}

\def\lrgldel {\mathchoice{(}{(}{\langle}{\langle}}%
\def\lrgrdel {\mathchoice{)}{)}{\rangle}{\rangle}}%
\def\aprldel {\mathchoice
    {\mathopen {\setbox0=\hbox{$\displaystyle     \lrgldel$}\hbox to\wd0
                         {\hfil$\displaystyle     (       $\hfil}}}%
    {\mathopen {\setbox0=\hbox{$\textstyle        \lrgldel$}\hbox to\wd0
                         {\hfil$\textstyle        (        $\hfil}}}%
    {\mathopen {\setbox0=\hbox{$\scriptstyle      \lrgldel$}\hbox to\wd0
                         {\hfil$\scriptstyle      (        $\hfil}}}%
    {\mathopen {\setbox0=\hbox{$\scriptscriptstyle\lrgldel$}\hbox to\wd0
                         {\hfil$\scriptscriptstyle(        $\hfil}}}}%
\def\aprrdel {\mathchoice
    {\mathclose{\setbox0=\hbox{$\displaystyle     \lrgrdel$}\hbox to\wd0
                         {\hfil$\displaystyle     )       $\hfil}}}%
    {\mathclose{\setbox0=\hbox{$\textstyle        \lrgrdel$}\hbox to\wd0
                         {\hfil$\textstyle        )        $\hfil}}}%
    {\mathclose{\setbox0=\hbox{$\scriptstyle      \lrgrdel$}\hbox to\wd0
                         {\hfil$\scriptstyle      )        $\hfil}}}%
    {\mathclose{\setbox0=\hbox{$\scriptscriptstyle\lrgrdel$}\hbox to\wd0
                         {\hfil$\scriptscriptstyle)        $\hfil}}}}%
\def\seqldel {\mathchoice
    {\mathopen {\setbox0=\hbox{$\displaystyle     \lrgldel$}\hbox to\wd0
                         {\hfil$\displaystyle     \langle  $\hfil}}}%
    {\mathopen {\setbox0=\hbox{$\textstyle        \lrgldel$}\hbox to\wd0
                         {\hfil$\textstyle        \langle  $\hfil}}}%
    {\mathopen {\setbox0=\hbox{$\scriptstyle      \lrgldel$}\hbox to\wd0
                         {\hfil$\scriptstyle      \langle  $\hfil}}}%
    {\mathopen {\setbox0=\hbox{$\scriptscriptstyle\lrgldel$}\hbox to\wd0
                         {\hfil$\scriptscriptstyle\langle  $\hfil}}}}%
\def\seqrdel {\mathchoice
    {\mathclose{\setbox0=\hbox{$\displaystyle     \lrgrdel$}\hbox to\wd0
                         {\hfil$\displaystyle     \rangle  $\hfil}}}%
    {\mathclose{\setbox0=\hbox{$\textstyle        \lrgrdel$}\hbox to\wd0
                         {\hfil$\textstyle        \rangle  $\hfil}}}%
    {\mathclose{\setbox0=\hbox{$\scriptstyle      \lrgrdel$}\hbox to\wd0
                         {\hfil$\scriptstyle      \rangle  $\hfil}}}%
    {\mathclose{\setbox0=\hbox{$\scriptscriptstyle\lrgrdel$}\hbox to\wd0
                         {\hfil$\scriptscriptstyle\rangle  $\hfil}}}}%
\def\parldel {\mathchoice
    {\mathopen {\setbox0=\hbox{$\displaystyle     \lrgldel$}\hbox to\wd0
                         {\hfil$\displaystyle     [       $\hfil}}}%
    {\mathopen {\setbox0=\hbox{$\textstyle        \lrgldel$}\hbox to\wd0
                         {\hfil$\textstyle        [        $\hfil}}}%
    {\mathopen {\setbox0=\hbox{$\scriptstyle      \lrgldel$}\hbox to\wd0
                         {\hfil$\scriptstyle      [        $\hfil}}}%
    {\mathopen {\setbox0=\hbox{$\scriptscriptstyle\lrgldel$}\hbox to\wd0
                         {\hfil$\scriptscriptstyle[        $\hfil}}}}%
\def\parrdel {\mathchoice
    {\mathclose{\setbox0=\hbox{$\displaystyle     \lrgrdel$}\hbox to\wd0
                         {\hfil$\displaystyle     ]       $\hfil}}}%
    {\mathclose{\setbox0=\hbox{$\textstyle        \lrgrdel$}\hbox to\wd0
                         {\hfil$\textstyle        ]        $\hfil}}}%
    {\mathclose{\setbox0=\hbox{$\scriptstyle      \lrgrdel$}\hbox to\wd0
                         {\hfil$\scriptstyle      ]        $\hfil}}}%
    {\mathclose{\setbox0=\hbox{$\scriptscriptstyle\lrgrdel$}\hbox to\wd0
                         {\hfil$\scriptscriptstyle]        $\hfil}}}}%

\def\parldeli #1{\mathchoice
    {\mathopen {\setbox0=\hbox{$\displaystyle     \lrgldel$}\hbox to\wd0
                         {\hfil$\displaystyle     \stackrel{#1}{[}       $\hfil}}}%
    {\mathopen {\setbox0=\hbox{$\textstyle        \lrgldel$}\hbox to\wd0
                         {\hfil$\textstyle        \stackrel{#1}{[}        $\hfil}}}%
    {\mathopen {\setbox0=\hbox{$\scriptstyle      \lrgldel$}\hbox to\wd0
                         {\hfil$\scriptstyle      \stackrel{#1}{[}        $\hfil}}}%
    {\mathopen {\setbox0=\hbox{$\scriptscriptstyle\lrgldel$}\hbox to\wd0
                         {\hfil$\scriptscriptstyle\stackrel{#1}{[}        $\hfil}}}}%
\def\parrdeli #1{\mathchoice
    {\mathclose{\setbox0=\hbox{$\displaystyle     \lrgrdel$}\hbox to\wd0
                         {\hfil$\displaystyle     \stackrel{#1}{]}       $\hfil}}}%
    {\mathclose{\setbox0=\hbox{$\textstyle        \lrgrdel$}\hbox to\wd0
                         {\hfil$\textstyle        \stackrel{#1}{]}        $\hfil}}}%
    {\mathclose{\setbox0=\hbox{$\scriptstyle      \lrgrdel$}\hbox to\wd0
                         {\hfil$\scriptstyle      \stackrel{#1}{]}        $\hfil}}}%
    {\mathclose{\setbox0=\hbox{$\scriptscriptstyle\lrgrdel$}\hbox to\wd0
                         {\hfil$\scriptscriptstyle\stackrel{#1}{]}        $\hfil}}}}%

\def\curldeli #1{\mathchoice
    {\mathopen {\setbox0=\hbox{$\displaystyle     \lrgldel$}\hbox to\wd0
                         {\hfil$\displaystyle     \stackrel{#1}{\{}       $\hfil}}}%
    {\mathopen {\setbox0=\hbox{$\textstyle        \lrgldel$}\hbox to\wd0
                         {\hfil$\textstyle        \stackrel{#1}{\{}        $\hfil}}}%
    {\mathopen {\setbox0=\hbox{$\scriptstyle      \lrgldel$}\hbox to\wd0
                         {\hfil$\scriptstyle      \stackrel{#1}{\{}        $\hfil}}}%
    {\mathopen {\setbox0=\hbox{$\scriptscriptstyle\lrgldel$}\hbox to\wd0
                         {\hfil$\scriptscriptstyle\stackrel{#1}{\{}        $\hfil}}}}%
\def\currdeli #1{\mathchoice
    {\mathclose{\setbox0=\hbox{$\displaystyle     \lrgrdel$}\hbox to\wd0
                         {\hfil$\displaystyle     \stackrel{#1}{\}}       $\hfil}}}%
    {\mathclose{\setbox0=\hbox{$\textstyle        \lrgrdel$}\hbox to\wd0
                         {\hfil$\textstyle        \stackrel{#1}{\}}        $\hfil}}}%
    {\mathclose{\setbox0=\hbox{$\scriptstyle      \lrgrdel$}\hbox to\wd0
                         {\hfil$\scriptstyle      \stackrel{#1}{\}}        $\hfil}}}%
    {\mathclose{\setbox0=\hbox{$\scriptscriptstyle\lrgrdel$}\hbox to\wd0
                         {\hfil$\scriptscriptstyle\stackrel{#1}{\}}        $\hfil}}}}%

\def\eightpoint{\small}                         
\def\pluldel {\mathchoice
   {\mathopen {\setbox0=\hbox{$\displaystyle     \lrgldel$}\hbox to\wd0
                        {\hfil$\displaystyle     [       $\hfil}%
                        \kern-\wd0\hbox to\wd0
                        {\hss$\vcenter{\hbox{\eightpoint$\scriptscriptstyle\bullet$}}$\hss}}}%
   {\mathopen {\setbox0=\hbox{$\textstyle        \lrgldel$}\hbox to\wd0
                        {\hfil$\textstyle        [       $\hfil}%
                        \kern-\wd0\hbox to\wd0
                        {\hss$\vcenter{\hbox{\eightpoint$\scriptscriptstyle\bullet$}}$\hss}}}%
   {\mathopen {\setbox0=\hbox{$\scriptstyle      \lrgldel$}\hbox to\wd0
                        {\hfil$\scriptstyle      [       $\hfil}%
                        \kern-\wd0\hbox to\wd0
                        {\hss$\raise.1ex\hbox{\eightpoint$\scriptscriptstyle\bullet$}$\hss}}}%
   {\mathopen {\setbox0=\hbox{$\scriptscriptstyle\lrgldel$}\hbox to\wd0
                        {\hfil$\scriptscriptstyle[       $\hfil}%
                        \kern-\wd0\hbox to\wd0
                        {\hss$\raise.03ex\hbox{\eightpoint$\scriptscriptstyle\bullet$}$\hss}}}}%
\def\plurdel {\mathchoice
   {\mathclose{\setbox0=\hbox{$\displaystyle     \lrgldel$}\hbox to\wd0
                        {\hfil$\displaystyle     ]       $\hfil}%
                        \kern-\wd0\hbox to\wd0
                        {\hss$\vcenter{\hbox{\eightpoint$\scriptscriptstyle\bullet$}}$\hss}}}%
   {\mathclose{\setbox0=\hbox{$\textstyle        \lrgldel$}\hbox to\wd0
                        {\hfil$\textstyle        ]       $\hfil}%
                        \kern-\wd0\hbox to\wd0
                        {\hss$\vcenter{\hbox{\eightpoint$\scriptscriptstyle\bullet$}}$\hss}}}%
   {\mathclose{\setbox0=\hbox{$\scriptstyle      \lrgldel$}\hbox to\wd0
                        {\hfil$\scriptstyle      ]       $\hfil}%
                        \kern-\wd0\hbox to\wd0
                        {\hss$\raise.1ex\hbox{\eightpoint$\scriptscriptstyle\bullet$}$\hss}}}%
   {\mathclose{\setbox0=\hbox{$\scriptscriptstyle\lrgldel$}\hbox to\wd0
                        {\hfil$\scriptscriptstyle]       $\hfil}%
                        \kern-\wd0\hbox to\wd0
                        {\hss$\raise.03ex\hbox{\eightpoint$\scriptscriptstyle\bullet$}$\hss}}}}%
\def\witldel {\mathchoice
   {\mathopen {\setbox0=\hbox{$\displaystyle     \lrgldel$}\hbox to\wd0
                        {\hfil$\displaystyle     (       $\hfil}%
                        \kern-\wd0\hbox to\wd0
                        {\hss$\vcenter{\hbox{\eightpoint$\scriptscriptstyle\bullet\mkern3.2mu$}}$\hss}}}%
   {\mathopen {\setbox0=\hbox{$\textstyle        \lrgldel$}\hbox to\wd0
                        {\hfil$\textstyle        (       $\hfil}%
                        \kern-\wd0\hbox to\wd0
                        {\hss$\vcenter{\hbox{\eightpoint$\scriptscriptstyle\bullet\mkern3.2mu$}}$\hss}}}%
   {\mathopen {\setbox0=\hbox{$\scriptstyle      \lrgldel$}\hbox to\wd0
                        {\hfil$\scriptstyle      (       $\hfil}%
                        \kern-\wd0\hbox to\wd0
                        {\hss$\raise.1ex\hbox{\eightpoint$\scriptscriptstyle\bullet\mkern3.2mu$}$\hss}}}%
   {\mathopen {\setbox0=\hbox{$\scriptscriptstyle\lrgldel$}\hbox to\wd0
                        {\hfil$\scriptscriptstyle(       $\hfil}%
                        \kern-\wd0\hbox to\wd0
                        {\hss$\raise.03ex\hbox{\eightpoint$\scriptscriptstyle\bullet\mkern3.2mu$}$\hss}}}}%
\def\witrdel {\mathchoice
   {\mathclose{\setbox0=\hbox{$\displaystyle     \lrgldel$}\hbox to\wd0
                        {\hfil$\displaystyle     )       $\hfil}%
                        \kern-\wd0\hbox to\wd0
                        {\hss$\vcenter{\hbox{\eightpoint$\scriptscriptstyle\mkern3.2mu\bullet$}}$\hss}}}%
   {\mathclose{\setbox0=\hbox{$\textstyle        \lrgldel$}\hbox to\wd0
                        {\hfil$\textstyle        )       $\hfil}%
                        \kern-\wd0\hbox to\wd0
                        {\hss$\vcenter{\hbox{\eightpoint$\scriptscriptstyle\mkern3.2mu\bullet$}}$\hss}}}%
   {\mathclose{\setbox0=\hbox{$\scriptstyle      \lrgldel$}\hbox to\wd0
                        {\hfil$\scriptstyle      )       $\hfil}%
                        \kern-\wd0\hbox to\wd0
                        {\hss$\raise.1ex\hbox{\eightpoint$\scriptscriptstyle\mkern3.2mu\bullet$}$\hss}}}%
   {\mathclose{\setbox0=\hbox{$\scriptscriptstyle\lrgldel$}\hbox to\wd0
                        {\hfil$\scriptscriptstyle)       $\hfil}%
                        \kern-\wd0\hbox to\wd0
                        {\hss$\raise.03ex\hbox{\eightpoint$\scriptscriptstyle\mkern3.2mu\bullet$}$\hss}}}}%

\newbox\ldelbox
\setbox\ldelbox=\hbox{$\lrgldel$}

\newbox\rdelbox
\setbox\rdelbox=\hbox{$\lrgrdel$}

\def\quadfs {\rlap{\rm\quad.}}%
\def\qquand {\qquad\mbox{and}\qquad}%
\def\clap#1{\hbox to 0pt{\hss#1\hss}}
\def\sqlap#1{\hbox to .5em{\hss#1\hss}}
\def\qlap#1{\hbox to 1em{\hss#1\hss}}
\def\qqlap#1{\hbox to 2em{\hss#1\hss}}
\def\qqqlap#1{\hbox to 3em{\hss#1\hss}}
\def\qqqqlap#1{\hbox to 4em{\hss#1\hss}}
\def\qqqqqlap#1{\hbox to 5em{\hss#1\hss}}
\def\qqqqqqlap#1{\hbox to 6em{\hss#1\hss}}
\def\qqqqqqqlap#1{\hbox to 7em{\hss#1\hss}}
\def\qqqqqqqqlap#1{\hbox to 8em{\hss#1\hss}}
\def\qqqqqqqqqlap#1{\hbox to 9em{\hss#1\hss}}
\newcommand{\wlap}[2][10ex]{\hbox to#1{\hss#2\hss}}

\newcommand{\wlapm}[2][10ex]{\hbox to#1{\hss$#2$\hss}}
\def\rlapm#1{\hbox to 0pt{$#1$\hss}}
\def\llapm#1{\hbox to 0pt{\hss$#1$}}

\def\qqquad{\quad\qquad}

\newcommand{\vclap}[2][0pt]{\hbox to #1{\hss#2\hss}}
\newcommand{\vclapm}[2][0pt]{\hbox to #1{\hss$#2$\hss}}

\def\interdisplayskip{.5ex}
\newskip\mydisplaywidth
\newcommand{\twolinedisplay}[3][10pt]{%
  \mydisplaywidth=\displaywidth
  \advance\mydisplaywidth-#1
  \begin{array}{c}
    \clap{\hbox to\mydisplaywidth{$\displaystyle#2$\hss}}\\[\interdisplayskip]
    \clap{\hbox to\mydisplaywidth{\hss$\displaystyle#3$}}
  \end{array}
}

\theoremstyle{plain}
\newtheorem{theorem}{Theorem}[section]

\newtheorem{lemma}[theorem]{Lemma} 
\newtheorem{proposition}[theorem]{Proposition}
\newtheorem{corollary}[theorem]{Corollary}
\newtheorem{conjecture}[theorem]{Conjecture}

\newtheorem{question}[theorem]{Question}

\theoremstyle{definition}

\newtheorem{example}[theorem]{Example}
\newtheorem{definition}[theorem]{Definition}
\newtheorem{notation}[theorem]{Notation}

\newtheorem{observation}[theorem]{Observation}
\newtheorem{remark}[theorem]{Remark}

\newcommand{\dual}[1]{\overline{#1}}
\def\square{\vcenter{\hbox{$\scriptstyle\oblong$}}}

\newcommand{\Ver}[1]{V({#1})} %
\newcommand{\Edg}[1]{E({#1})} %

\newcommand{\web}[1]{\mathcal{W}(#1)}

\newcommand{\critminmeas}{\nu}
\newcommand{\critmaxmeas}{\mu}
\newcommand{\redmeas}{e_{\vlan}}
\newcommand{\greenmeas}{e_{\vlor}}

\newcommand{\MedialCriterion}[2]{#1 \vartriangleleft\hspace{-1.9pt}\blacktriangleright #2 }

\newcommand{\numand}{\#_{\vlan}}
\newcommand{\numor}{\#_{\vlor}}
\begin{document}

\title[Linear rewriting systems for Boolean logic]{On linear rewriting systems for Boolean logic and some applications to proof theory}

\author[A.~Das]{Anupam Das\rsuper a}	%
\address{{\lsuper a}LIP, Universit\'e de Lyon, CNRS, ENS de Lyon, Universit\'e Claude-Bernard Lyon 1, Milyon}	%
\email{anupam.das@ens-lyon.fr}  %

\author[L.~Stra\ss burger]{Lutz Stra\ss burger\rsuper b}	%
\address{{\lsuper b}Inria Saclay, Palaiseau, France}	%
\email{lutz@lix.polytechnique.fr}  %
\titlecomment{{\lsuper*}This is an extended version of \cite{DasStra15:no-linear-sys} which appeared in the proceedings of \emph{RTA 2015}.}

\begin{abstract}
Linear rules have played an increasing role in structural proof theory in recent years. It has been observed that the set of all sound linear inference
rules in Boolean logic is already $\coNP$-complete, i.e.\ that
every Boolean tautology can be written as a (left- and right-)linear
rewrite rule. In this paper we study properties of systems consisting
only of linear inferences. Our main result is that the length of any
`nontrivial' derivation in such a system is bound by a polynomial. As a
consequence there is no polynomial-time decidable sound and complete
system of linear inferences, unless $\coNP=\NP$.
We draw tools and concepts from term rewriting, Boolean function
theory and graph theory in order to access some required intermediate
results. At the same time we make several connections between these
areas that, to our knowledge, have not yet been presented and
constitute a rich theoretical framework for reasoning about linear
TRSs for Boolean logic.
\end{abstract}

\maketitle

\section{Introduction}\label{sect:intro}

Consider the following conjunction rule from a Gentzen-style sequent calculus:
\begin{equation}
\label{eqn:conj-rule-gentzen}
\vliinf{}{}{\Gamma, A \wedge B, \Delta}{\Gamma, A \ }{\ B, \Delta}
\end{equation}
where $\Gamma$ and $\Delta$ are finite sequences of formulae. In this rule
all the formulae in the premisses occur in the conclusion with  the {same multiplicity}. In proof theory this is referred to as a \emph{multiplicative} rule.
This phenomenon can also be described as a \emph{linear} rule in term rewriting. For instance, the proof rule above has logical behaviour induced by the following linear term rewriting rule,
\begin{equation}
\label{eqn:conj-rule-trs}
(C \vlor A) \vlan(B \vlor D) \quad \to \quad C \vlor (A\vlan B) \vlor D
\end{equation}
where $C$ and $D$ here represent the disjunction of the formulae in $\Gamma$ and $\Delta $ respectively from~\eqref{eqn:conj-rule-gentzen}. 

This rule has been particularly important in structural proof theory, serving as the basis of Girard's \emph{multiplicative linear logic} \cite{Girard87}. A variant of~\eqref{eqn:conj-rule-trs}, that will play some role in this paper is the following,
\begin{equation}
\label{eqn:switch-rule}
\swi\quad:\quad A \vlan (B\vlor C) \quad \to \quad (A \vlan B) \vlor C
\end{equation}
which we call \emph{switch}, following \cite{Gugl:06:A-System:kl,GuglStra:01:Non-comm:rp,BrunTiu:01:A-Local-:mzF}, but which is also known as \emph{weak distributivity}~\cite{BCST}.

However the concept of linearity, or multiplicativity, itself is far more general. For instance, the advent of \emph{deep inference} has introduced the following linear rule, known as \emph{medial} \cite{BrunTiu:01:A-Local-:mzF}:
\begin{equation}
\label{eqn:medial-rule}
\med \quad  : \quad  (A \vlan B) \vlor (C \vlan D) 
\quad \to \quad 
(A \vlor C) \vlan (B \vlor D)
\end{equation}
This rule cannot be derived from \eqref{eqn:conj-rule-trs}, \eqref{eqn:switch-rule} or related rules, even when working modulo logical equivalence and logical constants.
From the point of view of proof theory \eqref{eqn:medial-rule} is particularly interesting since it allows for \emph{contraction},
\begin{equation}
\label{eqn:con-rule}
\gcd\quad:\quad A \vlor A\quad \to\quad A
\end{equation}
to be reduced to atomic form. For example consider the following transformation which reduces the logical complexity of a contraction step,
\begin{equation}
  \label{eqn:conred}
  \begin{array}{rl}
    \phantom{\annotatedarrow{}{\med}} & \underline{(A \vlan B) \vlor (A\vlan B)} \\
    \annotatedarrow{}{\gcd} & \ A \vlan B
  \end{array}
  \qquad \quad \leadsto  \quad \qquad
  \begin{array}{rl}
    \phantom{\annotatedarrow{}{\med}}& \underline{(A \vlan B) \vlor (A\vlan B)} \\
    \annotatedarrow{}{\med} & \underline{(A \vlor A)} \vlan (B\vlor B) \\
    \annotatedarrow{}{\gcd} & \ A \vlan \underline{(B\vlor B)} \\
    \annotatedarrow{}{\gcd} & \ A \vlan B
  \end{array}
\end{equation}
where redexes are underlined.

Until now the nature of linearity in Boolean logic has not been well understood, despite proving to be a concept of continuing interest in proof theory, cf.~\cite{Guglielmi:11:frogspost}, and category theory, cf.~\cite{str:medial,lamarche:gap}. While switch and medial form the basis of usual deep inference systems, it has been known for some time that other options are available: there are linear rules that cannot be derived from just these two rules (even modulo logical equivalences and constants), first explicitly shown in \cite{Stra:08:Extensio:kk}. The minimal known example, from \cite{Das:13:Rewritin:uq}, is the following:
\begin{equation}
\label{eqn:php32-linear}
\begin{array}{rl}
& ( A \vlor (B \vlan B') )
\vlan 
( ( C \vlan C' ) \vlor (D \vlan D') )
\vlan
((E \vlan E') \vlor F ) 
\\
\noalign{\smallskip}
\to & (A \vlan ( C \vlor E )) \vlor ( C' \vlan E' )
\vlor
(B' \vlan D') \vlor ( (B\vlor D) \vlan F )
\end{array}
\end{equation}
This example can be generalised to an infinite set of rules, where each rule is independent from all smaller rules.
In fact, the situation is rather more intricate than this: the set of linear inferences, denoted $\Lin$ henceforth, is itself $\coNP$-complete \cite{Stra:08:Extensio:kk}. This can be proved by showing that \emph{every} Boolean tautology can be written as a linear rule (which we demonstrate in Proposition~\ref{prop:coNP}).
This leads us to a natural question: 
\begin{question}
\label{quest:complete-basis-exists?}
Can we find a complete `basis' of linear inference rules? 
\end{question}
In other words, can proof theory {itself} be conducted in an entirely linear setting?
Such an approach would be in stark contrast with the traditional approach of \emph{structural} proof theory, which precisely emphasises the role of nonlinear behaviour via the structural rules, e.g.\ contraction and weakening. 

The main result of this work is a negative answer to the above question: there is no polynomial-time decidable linear TRS that is complete for $\Lin$, unless \hbox{$\coNP=\NP$}. Notice that the polynomial-time decidable criterion is essentially the most general condition one can impose without admitting a trivially positive answer to Question~\ref{quest:complete-basis-exists?} (e.g.\ by allowing the basis to be $\Lin$ itself). It is also a natural condition arising from proof theory, via the Cook-Reckhow definition of an abstract proof system \cite{CookReck:74:On-the-L:lr}. 

The high-level argument is as follows: 
\begin{enumerate}[label=(\Alph*)]
\item\label{item:nontriv-der-short} Any constant-free linear derivation of a `nontrivial' linear inference must have polynomial length.
\item\label{item:reduction-to-nontriv-frag} If a linear system is complete for $\Lin$ then arbitrary linear inferences can be derived from the `nontrivial' fragment of $\Lin$ with constant-free derivations of polynomial length.
\item\label{item:np-alg-for-conp} Putting these together, a complete linear system must admit polynomial-size derivations for any linear inference, inducing a $\NP$ algorithm for $\Lin$, and so $\coNP = \NP$.
\end{enumerate}
Point \eqref{item:nontriv-der-short} above represents the major technical contribution of this work. The proof requires us to work in three different settings: term rewriting, Boolean function theory and graph theory. Many of our intermediate results require elegant and novel interplays between these settings, taking advantage of their respective invariants; we try to make this evident in our exposition via various examples and discussion. 
Point \eqref{item:reduction-to-nontriv-frag} essentially appeared before in \cite{Das:13:Rewritin:uq}. We point out that the important point here is the \emph{existence} of small derivations, rather than the ability to explicitly construct them efficiently. 

Functions computed by linear terms of Boolean logic have been studied in Boolean function theory and circuit complexity for decades, where they are called ``read-once functions'' (e.g.\ in \cite{crama2011boolean}).\footnote{These have been studied in various forms and under different names. The first appearance we are aware of is in \cite{chein67:readonce-french}, and also the seminal paper of \cite{gurvich1977repetition} characterising these functions. The book we reference presents an excellent and comprehensive introduction to the area.} They are closely related to positional games (first mentioned in \cite{gurvich1982normal}) and have been used in the amplification of approximation circuits,  (first in \cite{Valiant1984363}, more generally in \cite{dubiner1997amplification}) as well as other areas.
However, despite this, it seems that there has been little study on \emph{logical} relationships between read-once functions, e.g.\ when one implies another. Many of the basic results and correspondences in this work, e.g.\ Proposition~\ref{prop:sound-minmax} and Theorem~\ref{thm:rel-minmax-webs}, have not appeared before, as far as we know, and themselves constitute interesting theoretical relationships.

This article is a full version of the extended abstract \cite{DasStra15:no-linear-sys}, which was presented at the \emph{RTA 2015} conference. In addition to providing full proofs for the various results, this version generally elaborates on many of the discussions in the previous version and gives a proof-theoretic context to this line of work. To this end we have included some further developments in Sections~\ref{sect:canon}, \ref{sect:norm-di} and \ref{sect:graph-logic} which are derived from our main result.

\bigskip

The structure of the paper is as follows. In Sections~\ref{sect:prelim-rew}, \ref{sect:webs} and \ref{sect:prelim-bool} we present preliminaries on each of our three settings and some basic results connecting various concepts between them. In Section~\ref{sect:inference} and \ref{sect:no-compl} we specialise to the setting of linear rewrite rules for Boolean logic and present our main results, Theorem~\ref{thm:main-result} and Corollary~\ref{cor:nolinsys}. 
In Sections~\ref{sect:canon} and \ref{sect:norm-di} we present some applications to deep inference proof theory, showing a form of \emph{canonicity} for medial and some general consequences for the normalisation of deep inference proofs. In Section~\ref{sect:graph-logic} we discuss a direction for future work in a graph-theoretic setting, and in Section~\ref{sect:conc} we present some concluding remarks, including relationships to models of linear logic and axiomatisations of Boolean algebras.

\subsubsection*{Acknowledgements}
We would like to thank Paola Bruscoli, Kaustuv Chaudhuri, Alessio Guglielmi, Willem Heijltjes and others in the deep inference community for many fruitful discussions on these topics. We would also like to thank the anonymous referees of this work and its previous versions for their useful comments.

\section{Preliminaries on rewriting theory}\label{sect:prelim-rew}

We work in the setting of first-order term rewriting as defined in the Terese textbook, \emph{Term Rewriting Systems} \cite{bezem2003term}. We will use the same notation for all symbols except the connectives, for which we use more standard notation from proof theory. 
In particular we will use $\fff$ and $\ttt$ for the truth constants, reserving $0$ and $1$ for the inputs and outputs of Boolean functions, introduced later.

We adopt one particular convention that differs from what is usual in the literature. A term rewriting system (TRS) is usually defined as an arbitrary set of rewrite rules.
	Here we insist that the set of instances of these rules, or reduction steps, is polynomial-time decidable.
The motivation is that we wish to be as general as possible without admitting trivial results. If we allowed all sets then a complete system could be specified quite easily indeed.
Furthermore, that an inference rule is easily or feasibly checkable is a usual requirement in proof theory, and in proof complexity this is formalised by the same condition on inference rules, cf.~\cite{CookReck:74:On-the-L:lr}.

\bigskip

Let us now consider Boolean logic in the term rewriting setting. Our language conists of the connectives
$\fff,\ttt,\wedge,\vee$ and a set $\Var$ of propositional variables, typically denoted $x,y,z$ etc. The set $\Var$ is equipped with an involution (i.e.\ self-inverse function) $\dual\cdot : \Var \to \Var$, such that $\bar x\neq x$ for all $x\in\Var$. We call $\bar x$ the \emph{dual} of $x$ and, for each pair of dual variables, we arbitrarily choose one to be \emph{positive} and the other to be \emph{negative}.

The set $\Ter$ of formulae, or \emph{terms}, is built freely from this
signature in the usual way. Terms are typically denoted by $s,t,u$ etc., and
term and variable symbols may occur with superscripts and subscripts if required.

In this setting $\ttt$ and $\fff$ are considered the constant
symbols of our language.
We say that a term $t$ is \emph{constant-free} if $\ttt$ and 
$\fff$ do not occur in $t$.

We do not include a symbol for negation in our language. This is due to the fact that soundness of a rewrite step is only preserved under \emph{positive} contexts. Instead we simply consider terms in negation normal form (NNF), which can be generated for arbitrary terms from positive and negative variables by the De Morgan laws:
\[
\dual \ttt = \fff
\qqquad
\dual \fff = \ttt
\qqquad
\bar{\bar x} = x
\qqquad
\dual{s\vlor t} = \bar s \vlan \bar t
\qqquad
\dual{s\vlan t} = \bar s \vlor \bar t
\]
We say that a term is \emph{negation-free} if it does not contain any negative variables.
We write $\Var(t)$ to denote the set of variables occurring in $t$.
We say that a term $t$ is \emph{linear} if, for each $x\in \Var(t)$, there is exactly one occurrence of $x$ in $t$.
The \emph{size} of a term $t$, denoted~$|t|$, is the total number of variable and function symbols occurring in $t$.
A \emph{substitution} is a mapping
$\sigma\colon\Var\to\Ter$ from the set of variables to the
set of terms such that $\sigma(x)\neq x$ for only finitely
many $x$.
The notion of substitution is extended to all terms, i.e.\ a map $\Ter \to \Ter$, in the usual way.       
A (one-hole) \emph{context} is a term with a single `hole' $\square$ occurring in place of a subterm. Below are three examples:
\begin{equation}
  \label{eq:ctx-ex}
C_1[\square ] := y \vlan (z \vlor \square )
\qquad
C_2[\square] := \square \vlor (w \vlan x)
\qquad
C_3[\square] := (w\vlan x)\vlor (y \vlan (z \vlor \square))
\end{equation}
We may write $C_{i} [t]$ to denote the term obtained by replacing the occurrence of $\square$ in $C_i [\square]$ with $t$. 
We may also replace holes with other contexts to derive new contexts. For example, notice that $C_3 [\square] $ in~\eqref{eq:ctx-ex} is equivalent, modulo commutativity of $\vee$, to $C_2 [C_1 [\square] ]$.

\begin{definition}[Rewrite rules]\label{def:rewrite}
	A \emph{rewrite rule} is an expression $l\to r$, where $l$ and $r$ are terms, such that $l\neq r$.
	We write $\rho: l \to r$ to express that the rule $l\to r$ is called $\rho$. 
	In this rule we call $l$ the \emph{left hand side (LHS)} of $\rho$, and $r$ its \emph{right hand side (RHS)}.
	We say that $\rho$ is \emph{left-linear} (resp.~\emph{right-linear}) if $l$ (resp.\ $r$) is a linear term.
	We say that $\rho$ is \emph{linear} if it is both left- and right-linear.
	We write $s \annotatedarrow{}{\rho}  t$ to express that $s\to t$ is a \emph{reduction step} of $\rho$, i.e.\ that $s = C[\sigma(l)]$ and $t = C[\sigma(r)]$ for some substitution $\sigma$ and some context $C[\square]$. 
\end{definition}

For instance, the rules $\swi $ from \eqref{eqn:switch-rule} and $\med $  and \eqref{eqn:medial-rule} are examples of linear rules. The rule $\gwu : x \vlan y \to x$ (which we consider later in Section~\ref{sect:norm-di}) is also linear, while the rule $\gcd $ from \eqref{eqn:con-rule} is not linear.

\begin{definition}
	[Term rewriting systems]
	The \emph{one-step} reduction relation of a set of rewrite rules $R$ is $\annotatedarrow{}{R}$, where $s\annotatedarrow{}{R} t$ if $s\annotatedarrow{}\rho t$ for some $\rho \in R$. A \emph{term rewriting system} (TRS) is a set of rewrite rules whose one-step reduction relation is decidable in polynomial time.
	A \emph{linear (term rewriting) system} is a TRS whose rules are all linear.
\end{definition}

\begin{definition}[Derivations]
	A \emph{derivation} under a binary relation
        $\annotatedarrow{}{R}$ on $\Ter$ is a finite sequence $\pi :
        t_0 \annotatedarrow{}{R} t_1 \annotatedarrow{}{R} \vldots
        \annotatedarrow{}{R} t_l$. The \emph{length} of $\pi$ is $l$.
	We also write $\annotatedarrow{*}{R}$ to denote the reflexive transitive closure of $\annotatedarrow{}{R}$.

	For an equivalence relation $\sim$ on $\Ter$ and a TRS $R$, we define  an \emph{$R$-derivation modulo $\sim$} as a sequence $\pi : t_0 \sim t_0' \annotatedarrow{}{R} t_1 \sim t_1' \annotatedarrow{}{R} \cdots \annotatedarrow{}{R} t_l \sim t_l'$. In this case we say that the length of $\pi$ is $l$, i.e.\ we do not count the $\sim $ steps. 
\end{definition}

We write $\AC$ to denote the smallest equivalence relation closed under contexts generated by
the following equations for associativity and commutivity of
$\wedge$ and $\vee$:
\begin{equation*}
  (x\vlan y)\vlan z = x\vlan (y\vlan z)
  \qquad
  (x\vlor y)\vlor z = x\vlor (y\vlor z)
  \qquad
  x\vlan y =  y\vlan x
  \qquad
  x\vlor y =  y\vlor x
\end{equation*}
Note that $\AC$ contains only linear equations. The following
equations for the constants are also linear and similarly generate a context-closed equivalence relation called $U$:
\[
x \vlor \fff = x = \fff \vlor x
\quad \quad
x\vlan \ttt = x = \ttt \vlan x
\quad \quad
\ttt \vlor \ttt = \ttt
\quad  \quad 
\fff \vlan \fff = \fff
\]
We denote by $\ACU$ the combined system of $\AC$ and $U$.
We will also need the system~$U'$ that extends~$U$ in the natural way by the following equations:
\[
x\vlor \ttt = \ttt = \ttt\vlor x \quad\quad x\vlan \fff = \fff = \fff\vlan x
\]
Notice that these are not linear in the sense of \cite{Das:13:Rewritin:uq}, but are considered linear in our more general setting. 
We denote by $\ACU'$ the combined system of $\AC$ and $U'$.

It turns out that this equivalence relation relates precisely those linear terms that compute the same Boolean function, as we will see later.

\section{Preliminaries on relation webs}\label{sect:webs}

In this section we restrict our attention to negation-free
constant-free linear terms and study their syntactic
structure, in the form of \emph{relation webs}
\cite{Gugl:06:A-System:kl,Stra:07:A-Charac:fk}.

We will consider graphs that are undirected, simple, and with labelled edges; we will make use of standard graph-theoretic terminology. 
For a graph $G$ we denote its \emph{vertex set} or set of \emph{nodes} as $\Ver{G}$, and the set of its \emph{labelled edges} as $\Edg{G}$. We say ``$\TwoGraphLabel[\smash{\star}]{x,y}b$ in $G$'' to express that the edge $\{x,y\}$ is labelled $\star$ in the graph $G$.
A~set~$X\subseteq \Ver{G}$ is a $\star$-\emph{clique} if every pair $x,y \in X$ has a $\star$-labelled edge between them. 
A \emph{maximal} $\star$-clique is a $\star$-clique that is not contained in any larger $\star$-clique.

Analysing the term tree of a negation-free constant-free linear term $t$, notice that for each pair of variables $x,y$ occurring in $t$, there is a unique connective $\star \in \{\wedge, \vee \}$ at the root of the smallest subtree containing the (unique) occurrences of $x$ and $y$. Let us call this the \emph{least common connective} of $x$ and $y$ in $t$.

\begin{definition}[Relation webs]\label{def:relweb}
	The \emph{(relation) web} $\web t$ of a constant-free negation-free linear term $t$ is the complete graph whose vertex set is $\Var(t)$, such that the edge between two variables $x$ and $y$ is labelled by their least common connective in $t$. We write $\redmeas(t)$
(resp.\ $\greenmeas(t)$) to be the number of $\wedge$-
(resp.\ $\vee$-)labelled edges in $\web t$.
\end{definition}

As a convention we will write $\redge[]{x}{y}$ if the edge
$\set{x,y}$ is labelled by $\cand$, and we write
$\gedge[]{x}{y}$ if it is labelled by~$\cor$.

\begin{example}
	\label{ex:linterm-web}
	The term $t=\vls[((v \vlor w).x).(y.z)]$ has the relation web: $$\FiveGraphR{v,x,y,z,w}rggggrgggr$$
        We have that $\redmeas(t)=3$ and $\greenmeas(t)=7$.
\end{example}

\begin{proposition}\label{prop:numredgreen}
  Let $t$ be a constant-free negation-free linear term with $n$
  variables, and let $e:= \frac{1}{2}n (n-1)$. Then
  $\redmeas(t),\greenmeas(t)\leq e$, and
  $\redmeas(t)+\greenmeas(t)=e$.
\end{proposition}

\begin{proof}
	This follows from the fact that there are only $e$ edges in a web, all of which must be labelled $\wedge$ or~$\vee$.
\end{proof}

\begin{remark}[Labels]
	We point out that, instead of using labelled complete graphs, we could have also used unlabelled arbitrary graphs, since we have only two connectives ($\wedge$ and $\vee$) and so one could be specified by the lack of an edge. This is indeed done in some settings, e.g.\ the cooccurrence graphs of \cite{crama2011boolean}. However, we use the current formulation in order to maintain consistency with the previous literature, e.g.\ \cite{Gugl:06:A-System:kl} and \cite{Stra:07:A-Charac:fk}, and since it helps write certain arguments, e.g.\ in Section~\ref{sect:canon}, where we need to draw graphs with incomplete information.
\end{remark}

One of the reasons for considering relation webs is the following proposition, which allows us to reason about equivalence classes modulo $\AC$ easily. 

\begin{proposition}\label{prop:web-mod-ac}
	Constant-free negation-free linear terms are equivalent modulo $\AC$ if and only if they have the same web.
\end{proposition}

\begin{proof}
  This follows immediately from the definition and that $\AC$ preserves least common connectives.
\end{proof}

An important property of webs is that they have no minimal paths of length $>2$. More precisely, we have the following:

\begin{proposition}
	A complete $\{\wedge,\vee \}$-labelled graph on $X$ is the web of some negation-free constant-free linear term on $X$ if and only if it contains no induced subgraphs of the form:
	\begin{equation}
	\label{eq:P4}
	\FourGraph{w,x,y,z}grrgrg
	\end{equation}
\end{proposition}

A proof of this property can be found, for example, in \cite{mohring:orders}, \cite{retore:phd}, \cite{bechet:etal:97}, or \cite{Gugl:06:A-System:kl}. It is called \emph{$P_4$-freeness} or \emph{Z-freeness} or \emph{N-freeness}, depending on the viewpoint. This property can be useful when we reason with webs, for instance in Section~\ref{sect:canon}.

\section{Preliminaries on Boolean functions}\label{sect:prelim-bool}
In this section we introduce the usual Boolean function models for terms of Boolean logic.
At the end of the section we give some examples of the various notions introduced.

A \emph{Boolean function} on a (finite) set of variables $X\subseteq \Var$ is a map $f\colon\{0,1\}^X \to \{0,1\}$.
We identify $\{0,1\}^X$ with $\mathcal P (X)$, the powerset of $X$,
i.e.\ we may specify an argument of a Boolean function by the subset of its variables assigned to $1$.
A little more formally, a function $\nu : X \to \{0,1\} $ is specified by the set
$X_\nu$ it indicates, i.e.\ $x\in X_\nu$ just if $\nu(x) = 1$.  For this reason we
may quantify over the arguments of a Boolean function by writing
$Y\subseteq X$ rather than $\nu \in \{0,1\}^X$, i.e.\ we write $f(Y)$
to denote the value of $f$ if the input is 1 for the variables in $Y$
and 0 for the variables in $X\setminus Y$. Similarly, we write
$f(\cmplmt Y)$ for the value of $f$ when the variables in $Y$ are $0$
and the variables in $X\setminus Y$ are $1$.

For Boolean functions $f,g : \{0,1\}^X \to \{0,1\}$ 
we write $f \leq g$ if, for every $ Y\subseteq X$, we have that $f(Y) \leq g(Y)$. 
Notice that the following can easily be shown to be equivalent:
\begin{enumerate}
	\item $f\leq g$.
	\item $f(Y) = 1 \implies g(Y) = 1$.
	\item $g(Y) = 0 \implies f(Y) = 0$.
\end{enumerate}
We also write $f<g$ if $f\leq g$ but $f(Y) \neq g(Y)$ for some $Y\subseteq X$.

\begin{definition}
  A Boolean function  $f\colon\{0,1\}^X \to \{0,1\}$ is \emph{monotone} iff  $Y \subseteq Y' \subseteq X $ implies $f(Y)\leq f(Y')$.
\end{definition}

\begin{definition}\label{def:minterm}
	Let $f$ be a monotone Boolean function on a variable set $X$. A set $Y\subseteq X$ is a \emph{minterm} (resp.~\emph{maxterm}) for $f$ if it is a minimal set such that $f(Y) =1$ (resp.\ \hbox{$f(\cmplmt{Y})=0$}). 
	The set of all minterms (resp.~maxterms) of $f$ is denoted $\MIN(f)$ (resp.\ $\MAX(f)$).
\end{definition}

\begin{observation}\label{obs:minmax}
  Monotone Boolean functions are uniquely determined by their minterms
  or by their maxterms. In particular, for two functions $f$ and $g$,
  we have $\MIN(f)\neq\MIN(g)$ iff $\MAX(f)\neq\MAX(g)$ iff there is a $Y$ such that
  $f(Y)\neq g(Y)$.
  
  We also have that, if $f(X) = 1$, then there is some $S \in \MIN (f)$  such that $S \subseteq X$; dually, if $f(X) =0 $, then there is some $T\in \MAX(f)$ such that $T \supseteq X$.
\end{observation}

Minterms and maxterms correspond to minimal DNF and CNF representations, respectively, of a monotone Boolean function. We refer the reader to \cite{crama2011boolean} for an introduction to their theory. In this work we use them in a somewhat different way to Boolean function theory, in that we devise definitions of logical concepts such as entailment and, in the next section, what we call ``triviality''. 
The reason for this is to take advantage of the purely function-theoretic results stated in this section (e.g.\ Gurvich's Theorem~\ref{thm:gurevich} below) to derive our main results in Sections~\ref{sect:inference} and \ref{sect:no-compl}.

\begin{proposition}\label{prop:sound-minmax} %
	For monotone Boolean functions $f,g$ on the same variable set, the following are equivalent:
	\begin{enumerate}
		\item\label{item:f-lessthan-g} $f \leq g$.
		\item\label{item:min-telescope} $\forall S \in \MIN(f) .\; \exists S' \in \MIN(g) .\; S' \subseteq S$.
		\item\label{item:max-telescope} $\forall T \in \MAX(g) .\; \exists T' \in \MAX(f) .\; T' \subseteq T$.
	\end{enumerate}
\end{proposition}

\begin{proof}
	$1\implies 2$. Suppose $f\leq g$ and let $S\in\MIN(f)$. We have that $f(S) = 1$ so also $g(S) =1$, by \ref{item:f-lessthan-g}, whence there must be an $S' \in \MIN(g)$ such that $S' \subseteq S$, by Observation~\ref{obs:minmax}.
	
	$2\implies 1$. If $f(X)=1$ then there is some $S \in \MIN(f)$ such that $S\subseteq X$, by Observation~\ref{obs:minmax}. By \ref{item:min-telescope}, there is some $S' \in \MIN(g)$ such that $S' \subseteq S$, and so $S'\subseteq X$. Therefore $g(X)=1$, by monotonicity, and so $f\le g$.
	
	$1\implies 3$ and $3\implies 1$ are proved similarly.
\end{proof}

A term $t$ computes a Boolean function $\{0,1\}^{\Var(t)} \to \{0,1\}$,
in the usual way, and negation-free terms compute monotone Boolean functions. Thus, we can speak of minterms and
maxterms of a negation-free term $t$, referring to the minterms and
maxterms of the function computed by $t$.
For linear terms, this will allow us to give a graph-theoretic formulation of minterms and
maxterms using concepts from the previous section. We
give the following inductive construction of minterms and maxterms:

\begin{proposition}\label{prop:altminterm}
	Let $t$ be a term.
	A set $S\subseteq \Var(t)$ is a minterm of $t$ if and only if:
	\begin{itemize}
		\item $t = \top$ and $S$ is empty, or
		\item $t = x$ and $S = \{x\}$, or
		\item $t=t_1 \vlor t_2$ and $S$ is a minterm of\/ $t_1$ or of\/ $t_2$, or
		\item $t=t_1 \vlan t_2$ and $S = S_1 \cup S_2$ where each $S_i$ is a minterm of $t_i$.
	\end{itemize}
	Dually, a set $T\subseteq \Var(t)$ is a maxterm of $t$ if and only if:
	\begin{itemize}
		\item $t = \bot$ and $T$ is empty, or
		\item $t=x$ and $T = \{x\}$, or
		\item $t=t_1 \vlor t_2$ and $T = T_1 \cup T_2$ where each $T_i$ is a maxterm of $t_i$, or
		\item $t=t_1 \vlan t_2$ and $T$ is a maxterm of\/ $t_1$ or of\/ $t_2$.
	\end{itemize}	
\end{proposition}

\begin{proof}
  This follows straightforwardly from Definition~\ref{def:minterm} and
  structural induction on~$t$.
\end{proof}

Notice that, in particular, $\bot$ has no minterms and $\top$ has no maxterms.
We can now present one of the important correspondences of this work, characterising minterms and maxterms of linear terms as maximal cliques in their relation webs:

\begin{theorem}\label{thm:rel-minmax-webs}
  A set of variables is a minterm (resp.~maxterm) of a negation-free
  constant-free linear term $t$ if and only if it is a maximal
  $\wedge$-clique (resp.\ maximal $\vee$-clique) in $\web t$.
\end{theorem}

\begin{proof}
  This follows from structural induction on $t$ and
  Proposition~\ref{prop:altminterm}.
\end{proof}

\begin{definition}[Read-once functions]
  A Boolean function is called  \emph{read-once} if it is computed by some
  linear term.
\end{definition}

It is not exactly clear when the following result first appeared,
although we refer to a discussion in \cite{crama2011boolean} where it
is stated that results directly implying this were first mentioned
in~\cite{kuznetsov1958non}. The result also occurs in
\cite{gurvich1977repetition}, and is generalised to certain other
bases in \cite{Heiman94onread-once}
and~\cite{Hellerstein:1990:CCL:895509}.

\begin{theorem}
	[Folklore]
	\label{thm:ro-mod-ac}
  Constant-free negation-free linear terms compute the same
  (read-once) Boolean function if and only if they are equivalent
  modulo $\AC$.
\end{theorem}

\begin{proof}
  This follows immediately from Proposition~\ref{prop:web-mod-ac},
  Theorem~\ref{thm:rel-minmax-webs}, and Observation~\ref{obs:minmax}.
\end{proof}

The following consequence of Theorem~\ref{thm:ro-mod-ac} appears
in \cite{Das:11:Depth-Change}, where a detailed proof may be found.

\begin{corollary}\label{cor:acu'}
	Negation-free linear terms compute the same (read-once) Boolean function if and only if they are equivalent modulo $\ACU'$.
\end{corollary}

\begin{proof}[Proof idea]
	The result essentially follows from the observation that every negation-free term is $\ACU'$-equivalent to $\fff$, $\ttt$ or a unique constant-free linear term. 
\end{proof}

Let us conclude this section by stating the following classical result, characterising the read-once functions over $\wedge$ and $\vee$,
due to Gurvich in \cite{gurvich1977repetition}. This has appeared in various presentations and, in particular, the proof appearing in \cite{crama2011boolean} uses `cooccurrence' graphs that correspond to our relation webs.

\begin{theorem}[Gurvich]\label{thm:gurevich}
	A monotone Boolean function $f$ is read-once if and only if 
	$$\forall S\in\MIN(f) .\; \forall T\in\MAX(f).\; |S\cap T| =1\quadfs$$
\end{theorem}

In this paper we will actually only need one direction of this
theorem: that for monotone read-once functions, minterms and
maxterms have singleton intersection. Using the different settings we have introduced, we arrive at a remarkably simple proof of this direction:
\begin{proof}
[Proof of left-right direction of Theorem~\ref{thm:gurevich}]
A minterm and maxterm of $f$ must intersect since, otherwise, we could simultaneously force $f$ to evaluate to $0$ and $1$. On the other hand, by Theorem~\ref{thm:rel-minmax-webs}, a minterm is a $\wedge$-maxclique of $\web{t}$ and a maxterm is a $\vee$-maxclique of $\web{t}$, and cliques with different labels can intersect at most once.
\end{proof}

This simple proof  exemplifies the usefulness of considering both
the graph theoretic viewpoint and the Boolean function viewpoint.
Such interplays will prove to be very useful in the remainder of this work.

\begin{example}
	\label{ex:bool-fn-summary}
	Consider the function computed by the term $t=\vls[((v \vlor w).x).(y.z)]$ from Example~\ref{ex:linterm-web}. Appealing to Proposition~\ref{prop:altminterm}, $t$ has minterms $\{v,x\}$, $\{w,x\}$ and $\{y,z\}$, and maxterms $\{v,w,y\}$, $\{ v,w,z\}$, $\{x,y\}$ and $\{x,z\}$.

	Now consider the Boolean `threshold' functions $\mathit{TH}^X_k : \{ 0,1\}^X \to \{0,1\}$, which return $1$ on just those $Y \subseteq X$ such that $|Y| \geq k$. 
	By defnition, this has minterms $S\subseteq X$ such that $|S| = k$ and maxterms $T\subseteq X$ such that $|T| = n-k + 1$.
	This means that for each minterm there is a maxterm that contains it or vice versa, depending on whether $k\geq \frac{|X|}{2}$.
	Therefore by Gurvich's result, Theorem~\ref{thm:gurevich}, $\mathit{TH}^X_k$ is read-once just when $k=1$, where it is computed by the disjunction of $X$, or when $k = |X| -1$, where it is computed by the conjunction of $X$.
	
	Now let $X = \{ v,w,x,y,z \}$. Appealing to Proposition~\ref{prop:sound-minmax}, we have that $t \leq \mathit{TH}^X_2$, since all minterms of $t$ have size $2$ and so are also minterms of $\mathit{TH}^X_2$. Dually, the maxterms of $\mathit{TH}^X_2$ are just the quartets of $X$, each of which contains some maxterm of $t$: if it does not contain $v$ or $w$ then it must contain both $\{x,y\}$ and $\{x,z\}$, if it does not contain $x$ then it must contain both $\{v,w,y\} $ and $\{v,w,z\}$, and if it does not contain $y$ (or $z$) then it must contain both $\{v,w,z\}$ and $\{x,z\}$ (respectively $\{v,w,y\}$ and $\{x,y\}$).
\end{example}

\section{Linear inferences, triviality and a polynomial bound on length}\label{sect:inference}

In the previous section we considered the semantics of linear terms
via Boolean functions. In this section we study sound rewriting steps
between linear terms, with respect to this semantics, and prove our main result, Theorem~\ref{thm:main-result}, about the length of
such rewriting paths, corresponding to point \eqref{item:nontriv-der-short} in the Introduction, Section~\ref{sect:intro}.

\begin{definition}[Soundness]
  We say that a rewrite rule $s \to t$ is \emph{sound} if $s$ and $t$
  compute Boolean functions $f$ and $g$, respectively, such that
  $f\leq g$.  We say that a TRS is sound if all its rules are sound.
  A \emph{linear inference} is a sound linear rewrite rule.  
\end{definition}

\begin{notation}
  To switch conveniently between the settings of terms and Boolean
  functions, we freely interchange notations, e.g.\ writing $s\leq t$
  to denote that $s\to t $ is sound, and saying $f\to g$ is sound when
  $f\leq g$.
\end{notation}

We immediately have the following, which can also be found
in~\cite{Das:13:Rewritin:uq}.

\begin{proposition}
  Any sound negation-free linear TRS, modulo $\ACU'$, is terminating in exponential-time.\footnote{Strictly speaking, we mean that any derivation can be `expressed' as one of exponential length: if either associativity or commutativity is in the TRS then we could pathologically create arbitrarily long derivations.}
\end{proposition}

\begin{proof}
  The result follows by Boolean semantics and
  Corollary~\ref{cor:acu'}: each consequent term must compute a
  distinct Boolean function that is strictly bigger, under $\leq$, and
  the graph of $\leq$ has length $2^n$, where $n$ is the number of
  variables in the input term.
\end{proof}

The purpose of this section is now to put a polynomial bound on the
length of certain linear derivations. For this, the fundamental concept we use
is that of ``triviality'', first introduced in
\cite{Das:13:Rewritin:uq} as ``semantic triviality''.

\begin{definition}[Triviality]
	Let $f$ and $g$ be Boolean functions on a set of variables $X$, and let $x\in X$. We say $\fnimplies{f}{g}$  is \emph{trivial} at $x$ if for all $Y\subseteq X$, we have $f(Y\cup\set{x}) \leq g(Y\setminus\set{x})$. We say simply that $f\to g$ is \emph{trivial} if it is trivial at one of its variables.
\end{definition}

The idea behind triviality of a variable in an inference is that the validity of the inference is ``independent'' of the behaviour of that variable. 

\begin{example}
	\label{ex:triviality}
	Recalling the Boolean threshold functions $\mathit{TH}^X_k$ from Example~\ref{ex:bool-fn-summary}, notice that $\mathit{TH}^X_{k+1} \to \mathit{TH}^X_{k}$ is trivial at any (but at most one) variable of $X$.
	More concretely, the linear inference
	$
	x \vlan y 
	 \to 
	x \vlor y
	$
	is trivial at $x$ or $y$, whereas the linear inference,
	\begin{equation}
	\label{eqn:smix}
	x \vlan (y_1 \vlor \vldots \vlor y_n )
	\quad \to \quad
	x \vlor (y_1 \vlan \vldots \vlor y_n)
	\end{equation}
	is trivial at all $y_i$ simultaneously.
\end{example}

As observed in \cite{Das:13:Rewritin:uq}, the inference \eqref{eqn:smix} above can be used to create exponential-length (constant-free) linear derivations. 
The idea is to construct a derivation from the conjunction of a variable set $X$ to its disjunction, by induction on $|X|$, as follows,
\[
\begin{array}{rc}
& \vls(x . \underline{(y_1 \vlan \vldots  \vlan y_n)}) \\
\to &   \\
\vdots & \\
\to & \underline{x \vlan (y_1 \vlor \vldots \vlor y_n)} \\
\annotatedarrow{}{} & x \vlor \underline{(y_1 \vlan \vldots \vlan y_n)}\\
\to & \\
\vdots & \\
\to & \vls[x . (y_1 \vlor \vldots \vlor y_n)]
\end{array}
\]
where redexes are underlined and the two intermediate derivations are obtained from the inductive hypothesis.
We will show in the remainder of this section that such exponential length rewrite paths \emph{only} occur when deriving a triviality.

\begin{remark}[Hereditariness of triviality]\label{rmk:hered}
	Notice that the triviality property is somehow hereditary: if a sound sequence $f_0\to f_1\to\ldots \to f_l$ of Boolean functions is trivial at some point $f_i\to f_{i+1}$ for $0\le i<l$ then $f_1\to f_l$ is trivial.  However the converse does not hold: if the first and last function of a sound sequence constitutes a trivial pair it may be that there is no local triviality in the sequence. For example the endpoints of the derivation,
	\begin{equation}
	\label{eqn:triv-not-local}
	(w \vlan x) \vlor (y \vlan z) \quad \to \quad (w \vlor y) \vlan (x \vlor z)  \quad \to     \quad w \vlor x \vlor (y \vlan z)
	\end{equation}
	form a pair that is trivial at $w$ (or trivial at $x$), but no local step witnesses this. In these cases we call the sequence \emph{globally} trivial. This phenomenon is what we will need to address later in Lemma~\ref{intersectionlemma}, on which our main result crucially relies.
\end{remark}

In a similar way to how we expressed soundness via minterms or maxterms in Proposition~\ref{prop:sound-minmax}, we can also define triviality via minterms or maxterms.

\begin{proposition}\label{prop:trivial-minmax}
	The following are equivalent:
	\begin{enumerate}
		\item\label{item:trivial-inference} $\fnimplies{f}{g}$ is trivial at $x$.
		\item\label{item:trivial-min} $\forall S \in \MIN(f) .\; \exists S' \in \MIN(g) .\; S' \subseteq S\setminus\{x\}$.
		\item\label{item:trivial-max} $\forall T \in \MAX(g) .\; \exists T' \in \MAX(f) .\; T' \subseteq T\setminus\{x\}$.
	\end{enumerate}
\end{proposition}

\begin{proof}
	We first show that $\ref{item:trivial-inference}\implies \ref{item:trivial-min}$. Assume $\fnimplies{f}{g}$ is trivial at $x$, and let $S\in\MIN(f)$. We have $f(S)=1$, and hence also $f(S\cup\set{x})=1$. By way of contradiction assume there is no $S' \in \MIN(g)$ with $S' \subseteq S\setminus\set{x}$. Therefore $g(S\setminus\set{x})=0$, by Observation~\ref{obs:minmax}, contradicting triviality at~$x$. Next, we show $\ref{item:trivial-min}\implies \ref{item:trivial-inference}$. For this, let $Y$ be such that $f(Y\cup\set{x})=1$. Then there is a minterm $S\in\MIN(f)$ with $S\subseteq Y\cup\set{x}$, by Observation~\ref{obs:minmax}. By \ref{item:trivial-min}, there is a minterm $S' \in \MIN(g)$ with $S' \subseteq S\setminus\set{x}$. Hence $S'\subseteq Y\setminus\set{x}$ so $g(Y\setminus\set{x})=1$, by monotonicity, and thus $\fnimplies{f}{g}$ is trivial at~$x$. We prove $\ref{item:trivial-inference}\implies \ref{item:trivial-max}\implies \ref{item:trivial-inference}$ analogously.
\end{proof}

Let us now fix a sequence $f = f_0 < f_1 < \vldots < f_l = g$ of
strictly increasing read-once Boolean functions on a variable set $X$.
Intuitively, we would like to build a decreasing chain of minterms, whence we could extract an appropriate bound for $l$.
The problem, however, is that new minterms can appear too, for example in the case of medial \eqref{eqn:medial-rule},
so this process does not clearly terminate in reasonable time.

To address this issue, we will show that there must exist particular chains of minterms, for each variable, which will strictly decrease sufficiently often. Unless $f\to g$ is trivial, 
for each variable $x\in X$ we must be able to associate a minterm $S^x$ of $f$ such that, 
for any $S\subseteq S^x$ that is a minterm of some $f_i$, it must be that $S\ni x$.
This is visualized in Figure~\ref{fig:intersectionlemma} together with the dual property for maxterms.

\begin{figure}[!t]
		\[
		\begin{tikzpicture}
		\node at (-3.5,2.75) {$f_0$};
		\node at (-3.5,0) {$x$};
		\node at (-2.25,2.75) {$<$};
		\node at (-2.25,-0.75) {$\subseteq$};
		\node at (-2.25,0.75) {$\supseteq$};
		\node at (-1,2.75) {$f_1$};
		\node at (-1,0) {$x$};
		\node at (0.25,2.75) {$<$};
		\node at (0.25,-0.75) {$\subseteq$};
		\node at (1.5,2.75) {$\vldots$};
		\node at (0.25,0.75) {$\supseteq$};
		\node at (1.5,0) {$\vldots$};
		\node at (2.75,2.75) {$<$};
		\node at (2.75,-0.75) {$\subseteq$};
		\node at (2.75,0.75) {$\supseteq$};
		\node at (4,2.75) {$f_l$};
		\node at (4,0) {$x$};
		\draw[r,rounded corners] (-3,0.5) .. controls (-3,0) and (-3.5,-1) .. (-4,0.5) .. controls (-4,1) and (-4.5,1.5) .. (-3.5,2) .. controls (-2.5,2) and (-3,1) .. (-3,0.5);
		\draw[rounded corners,green,dashed,thick] (-3.8,-0.2) .. controls (-3.8,0.3) and (-3.3,0.3) .. (-2.8,-0.2) .. controls (-2.8,-0.3) and (-2.8,-0.3) .. (-3.3,-1.2) .. controls (-3.8,-1.2) and (-3.8,-0.7) .. (-3.8,-0.2);
		\draw[r,rounded corners] (-0.5,0.5) .. controls (-0.5,0) and (-1,-1) .. (-1.5,0.5) .. controls (-1.5,1) and (-1.5,1.5) .. (-1,2) .. controls (0,2) and (-0.5,1) .. (-0.5,0.5);
		\draw[r,rounded corners] (4.5,0.5) .. controls (4.5,0) and (4,-1) .. (3.5,0.5) .. controls (3.5,0.5) and (3.5,1) .. (4,1) .. controls (4,0.5) and (4.5,1) .. (4.5,0.5);
		\draw[rounded corners,green,dashed,thick] (-1.3,-0.2) .. controls (-1.3,0.3) and (-0.8,0.3) .. (-0.3,-0.2) .. controls (-0.3,-0.3) and (-0.3,-0.8) .. (-0.8,-1.7) .. controls (-1.3,-1.2) and (-1.3,-0.7) .. (-1.3,-0.2);
		\draw[green,thick,dashed,rounded corners] (3.5,0) .. controls (3.5,0.5) and (4,0.5) .. (4.5,0) .. controls (4.5,-0.5) and (5,-2) .. (4,-2.5) .. controls (3,-1.5) and (3.5,-0.5) .. (3.5,0);
		\end{tikzpicture}
		\]
		\caption{The critical minterms and maxterms of a sound sequence, cf.\ Lemma~\ref{intersectionlemma}.}
                \label{fig:intersectionlemma}
\end{figure}
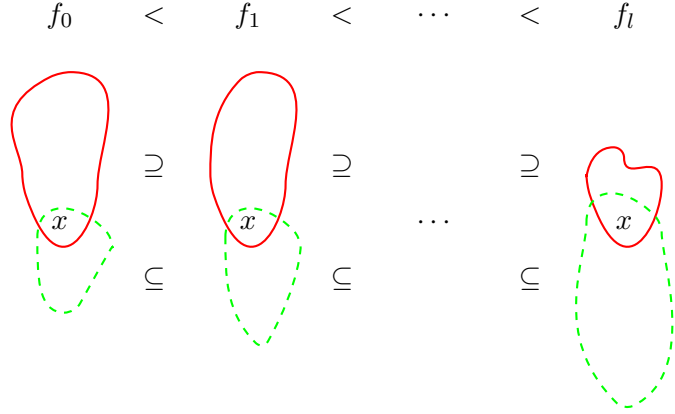

\begin{lemma}[Subset and intersection lemma]\label{intersectionlemma}
	Suppose $f\to g$ is not trivial.
	For every variable $x\in X$, there is a minterm $S^x$ of $f$ and a maxterm $T^x$ of $g$ such that:
	\begin{enumerate}
		\item\label{subsetmintermsitem} $\forall S_i \in \MIN(f_i). (S_i \subseteq S^x \implies x\in S_i)$.
		\item\label{subsetmaxtermsitem} $\forall T_i \in \MAX(g_i). (T_i \subseteq T^x \implies x\in T_i)$.
		\item\label{intersectionitem} $\forall S_i \in \MIN(f_i) . \forall T_i \in \MAX(g_i) . (S_i \subseteq S^x, T_i \subseteq T^x \implies S_i \cap T_i = \{x\})$.
	\end{enumerate}
\end{lemma}
\begin{proof}
	Suppose that, for some variable $x$ no minterm of $f$ has property \ref{subsetmintermsitem}. In other words, for every minterm $S^x$ of $f$ containing $x$ there is some minterm $S_i$ of some $f_i $ that is a subset of $S^x$ yet does not contain $x$. Since $\fnimplies{f_i}{f_l}$ is sound for every $i$ we have that, by Proposition~\ref{prop:sound-minmax}, for every minterm $S^x$ of $f$ containing $x$ there is some minterm $S_l$ of $f_l =g $ that is a subset of $S^x$ not containing $x$. I.e.\ $\fnimplies{f}{g}$ is trivial, by Proposition~\ref{prop:trivial-minmax}, which is a contradiction.
	Property \ref{subsetmaxtermsitem} is proved analogously.
	Finally, Property \ref{intersectionitem} is proved by appealing to read-onceness: any such $S_i$ and $T_i$ must contain $x$ by properties \ref{subsetmintermsitem} and \ref{subsetmaxtermsitem}, yet their intersection must be a singleton by Theorem~\ref{thm:gurevich} since all $f_i$ are read-once.
\end{proof}

Notice that, since some such $S_i$ and $T_i$ must exist for all $i$, by soundness, we can build a chain
of such minterms and maxterms preserving the intersection point. For a given derivation, let us call a choice of such minterms and maxterms \emph{critical} (see Figure~\ref{fig:intersectionlemma}).

We now state the main result of this section, also the main technical contribution of this work, for which Lemma~\ref{intersectionlemma} will play a crucial role and from which we can obtain our further results. While we state this result for terms, in order to access simultaneously the notions of relation webs and Boolean semantics, this could equally be stated in the setting of read-once Boolean functions due to Gurvich's result, Theorem~\ref{thm:gurevich}.

\begin{theorem}\label{thm:main-result}
  Let $s=t_0 < t_1 < \vldots < t_l =t$ be a (strictly increasing under $\leq$) sequence of negation-free constant-free linear terms on variable set $X$ of size $n$,
  such that $l>0$ and such that $s\to t$ is not trivial. We have that $l = O(n^4)$.
\end{theorem}

The remainder of this section is devoted to the proof of
Theorem~\ref{thm:main-result}. For this let us fix $\pi$ to denote the sequence $s=t_0 < t_1 < \vldots < t_l =t$. 
Recall that, since $t_i < t_{i+1}$, $t_i $ and $t_{i+1}$ have distinct minterms and maxterms, by Observation~\ref{obs:minmax}, and so must have distinct relation webs by Theorem~\ref{thm:rel-minmax-webs}.

We now
fix, for each $x\in X$ and $0\le i\le l$, some choice of $S^x_i$ and $T^x_i$ as
critical minterms and maxterms, respectively, of $t_i$, under
Lemma~\ref{intersectionlemma}.  I.e.\ we have that, for each $x \in
X$:
\begin{enumerate}
	\item\label{item:intfixedcrit} $S^x_i \cap T^x_i = \{x\}$ for each $i \leq l$.
	\item\label{item:minfixedcrit} $S^x_0 \supseteq S^x_1 \supseteq \vldots \supseteq S^x_l$.
	\item\label{item:maxfixedcrit} $T^x_0 \subseteq T^x_1 \subseteq \vldots \subseteq T^x_l$.
\end{enumerate}
We denote the size of the critical minterms and maxterms of $t_i$ by $|S^x_i|$ and $|T^x_i|$, respectively. Now we define:  
\begin{equation}
  \critminmeas(t_i) := \sum\limits_{x\in X} |S^x_i|
  \qquand
  \critmaxmeas(t_i) := \sum\limits_{x\in X} |T^x_i|
\end{equation}

\begin{observation}\label{obs:critmin}
  Note that we always have $|S^x_i|,|T^x_i| \leq n$ because a minterm
  or maxterm is a subset of $X$, and therefore we have
  $\critminmeas(t_i) , \critmaxmeas(t_i) \leq n^2$ for all $t_i$ in
  $\pi$.
\end{observation}

The following two propositions now form the core of the argument. The
first says that whenever a $\wedge$-edge changes to a $\vee$-edge,
some minterm strictly decreases in size, and the second one says that
if a minterm strictly decreases in size then some critical maxterm must
strictly increase in size. 
Thus the proof of Theorem~\ref{thm:main-result} that follows again relies crucially on the interplay between the Boolean function setting and the graph-theoretic setting.

\begin{figure}
	\[
	\begin{tikzpicture}
	\node at (-8.5,1) {$\web{t_i}:$};
	\node at (-2.5,1) {$\to$};
	\node at (-1,1) {$\web{t_{i+1}}:$};
	\node (v2) at (-4,0) {$x$};
	\node (v4) at (4,0) {$x$};
	\node (v1) at (-6,2) {$y$};
		\draw[r]  (v1) edge (v2);
		\node (v3) at (2,2) {$y$};
		\draw[green,thick,dashed]  (v3) edge (v4);
		\node at (-7,1.5) {$S$};
		\draw[r,rounded corners] (-6,2.5) .. controls (-7,2.5) and (-8,1.5) .. (-7,1) .. controls (-6.5,1) and (-5.5,1) .. (-5,0.5) .. controls (-4.5,-0.5) and (-3.5,-0.5) .. (-3.5,0.5) .. controls (-3.5,1.5) and (-4.5,2.5) .. (-6,2.5);
		\draw[r,rounded corners] (2,2.5) .. controls (1,2.5) and (0,1.5) .. (1,1) .. controls (2.5,1) and (3.5,0.5) .. (4,1) .. controls (4,1.5) and (3.5,2.5) .. (2,2.5);
		\node at (1,1.5) {$S'$};
	\end{tikzpicture}
	\]
	\caption{In the proof of
          Proposition~\ref{prop:minterm-decrease}, $S'$ cannot contain
          both $x$ and $y$, so we can assume without loss of
          generality that it does not contain $x$ (although it need not necessarily contain $y$ either).}
        \label{fig:minterm-decrease}
\end{figure}
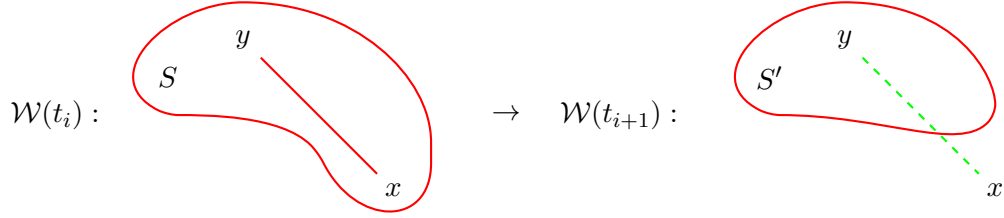

\begin{proposition}\label{prop:minterm-decrease}
	Suppose, for some $i<l$, we have that $\redge[]xy$ in $\web{t_i}$ and $\gedge[]xy$ in $\web{t_{i+1}}$. 
	Then there is a minterm $S$ of $t_i$, and a minterm $S'$ of $t_{i+1}$ such that $S' \subsetneq S$.
\end{proposition}

\begin{proof}
	Take any maximal $\wedge$-clique in $\web{t_i}$ containing $x$ and $y$, 
	of which there must be at least one. 
	This must have a $\wedge$-subclique which is maximal in $\web{t_{i+1}}$, by
	Proposition~\ref{prop:sound-minmax} and Theorem~\ref{thm:rel-minmax-webs}. 
	This subclique cannot contain both $x$ and $y$, so the inclusion must be strict (see Figure~\ref{fig:minterm-decrease}).
\end{proof}

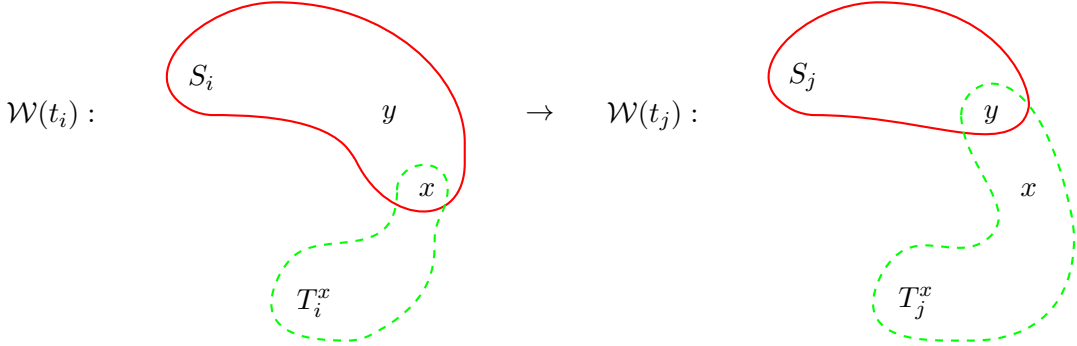
\begin{figure}
	\[
	\begin{tikzpicture}
	\node at (-9,1) {$\web{t_i}:$};
	\node at (-2.5,1) {$\to$};
	\node at (-1,1) {$\web{t_j}:$};
	\node (v2) at (-4,0) {$x$};
	\node (v4) at (4,0) {$x$};
		\node at (-7,1.5) {$S_i$};
		\draw[r,rounded corners] (-6,2.5) .. controls (-7,2.5) and (-8,1.5) .. (-7,1) .. controls (-6.5,1) and (-5.5,1) .. (-5,0.5) .. controls (-4.5,-0.5) and (-3.5,-0.5) .. (-3.5,0.5) .. controls (-3.5,1.5) and (-4.5,2.5) .. (-6,2.5);
\node at (1,1.5) {$S_j$};
		\draw[r,rounded corners] (2,2.5) .. controls (1,2.5) and (0,1.5) .. (1,1) .. controls (2.5,1) and (3.5,0.5) .. (4,1) .. controls (4,1.5) and (3.5,2.5) .. (2,2.5);
		\node at (2.5,-1.5) {$T^x_{j}$};
		\draw[green,dashed,thick,rounded corners] (3.1,1) .. controls (3.1,1.5) and (4.1,2) .. (4.6,-0.5) .. controls (4.6,-1.5) and (4.1,-2) .. (3.1,-2) .. controls (2.6,-2) and (1.6,-2) .. (2.1,-1) .. controls (2.6,-0.5) and (3.1,-1) .. (3.6,-0.5) .. controls (3.6,0) and (3.1,0.5) .. (3.1,1);
\node at (3.5,1) {$y$};
		\node at (-5.5,-1.5) {$T^x_i$};
		\draw[green,dashed,thick,rounded corners] (-4.4,0) .. controls (-4.4,0.5) and (-3.4,0.5) .. (-3.9,-0.5) .. controls (-3.9,-1) and (-3.9,-1.5) .. (-4.9,-2) .. controls (-5.4,-2) and (-6.4,-2) .. (-5.9,-1) .. controls (-5.4,-0.5) and (-4.4,-1) .. (-4.4,0);
\node at (-4.5,1) {$y$};
	\end{tikzpicture}
	\]
	\caption{If some minterm becomes smaller then some critical maxterm must become bigger.}
        \label{fig:mindec-maxinc}
\end{figure}

\begin{proposition}\label{prop:mindec-maxinc}
	Suppose for $j>i$ there is some minterm $S_i$ of $t_i$ and some minterm $S_j$ of $t_j$ such that $S_j\subsetneq S_i$. 
	Then, for some variable $x\in X$, we have that $T^x_i \subsetneq T^x_j$.
\end{proposition}
\begin{proof}
	We let $x$ be some variable in $x\in S_i \setminus S_j $,
        which must be nonempty by hypothesis.  By
        Theorem~\ref{thm:gurevich} we have that $|T^x_i \cap S_i |=1$,
        so it must be that $T^x_i \cap S_i = \{x\}$ by construction.
	On the other hand we also have that $|T^x_j \cap S_j| = 1$,
        and so there is some (unique) $y \in T^x_j \cap S_j$.  Now,
        since $S_i \supsetneq S_j$ we must have $y\in S_i$.  However
        we cannot have $y \in T^x_i$ since that would imply that
        $\{x,y\} \subseteq T^x_i \cap S_i$, contradicting the above.
        Since we have that $T^x_i \subseteq T^x_j$ we can now
        conclude that $T^x_i \subsetneq T^x_j$ as required, because $y
        \in T^x_j$ and $y \notin T^x_i$ (see Figure~\ref{fig:mindec-maxinc}).
\end{proof}

Notice that both of the two propositions above rely crucially on the notion of linearity. Proposition~\ref{prop:minterm-decrease} assumes the existence of relation webs for a term, a property peculiar to linear terms, whereas Proposition~\ref{prop:mindec-maxinc} does not remain true for terms that do not compute read-once Boolean functions: there is no requirement for minterms and maxterms of arbitrary Boolean functions to intersect at most once, cf.~Example~\ref{ex:bool-fn-summary}.

\begin{lemma}[Increasing measure] 
  \label{lem:increasingmeasure}
  The lexicographical product $\critmaxmeas \times \redmeas$ is
  strictly increasing at each step of $\pi$.
\end{lemma}
\begin{proof}
  Notice that, by
  Lemma~\ref{intersectionlemma}.\ref{subsetmaxtermsitem}, we have that
  $T^x_0 \subseteq T^x_1 \subseteq \vldots \subseteq T^x_l$, which means that
  $\critmaxmeas$ is non-decreasing.  So let us consider the case that
  $\redmeas$ decreases at some step and show that $\critmaxmeas$ must
  strictly increase. If $\redmeas(t_i) > \redmeas(t_{i+1}) $ then we
  must have that some edge is labelled $\wedge$ in $\web{t_i}$ and
  labelled $\vee$ in $\web{t_{i+1}}$. Hence, by
  Proposition~\ref{prop:minterm-decrease} some minterm has strictly
  decreased in size and so by Proposition~\ref{prop:mindec-maxinc} some
  critical maxterm must have strictly increased in size.
\end{proof}

From here we can finally prove our main result.

\begin{proof}
  [Proof of Theorem~\ref{thm:main-result}] By
  Observation~\ref{obs:critmin} and Proposition~\ref{prop:numredgreen}
  we have that $\critmaxmeas = O(n^2) = \redmeas$ and so, since $s\to
  t$ is nontrivial, it must be that the length $l$ of $\pi$ is
  $O(n^4)$, as required.
\end{proof}

Notice that, while the various settings exhibit a symmetry between
$\wedge$ and $\vee$, it is the property of soundness that induces the
necessary asymmetry required to achieve this result.

\begin{remark}
Let us take a moment to reflect on what might happen if the inference that is derived were trivial. Consider the following:
	\[
	\vls(w.x.(y\vlor z))
	\quad \to \quad
	\vls(w.((x\vlan y) \vlor z))
	\quad \to \quad
	\vls(w.(x\vlor y\vlor z))
	\]
	This derivation is trivial at $x$, in fact witnessed by the second inference.\footnote{Although notice we could equally consider a (globally) trivial derivation with no local triviality if, say, $z$ were replaced by a conjunction $z_1 \vlan z_2$, appealing to Remark~\ref{rmk:hered} and using \eqref{eqn:triv-not-local} to derive the second step.}	
Notice that there is no `critical' minterm for $y$ in this derivation: the only minterm containing $y$ on the left is $\{w,x,y\}$, but this contains a minterm $\{w,x\}$ on the right. This is similarly true for $z$, although here the situation is rather worse: while the minterm $\{w,x,z\}$ on the left indeed contains $\{w,x\}$ on the right, there is no intermediate minterm. This prevents us from proving termination via a step-by-step analysis of the subsets of $\{w,x,z\}$ that occur as minterms in the derivation, which we are able to do in the presence of critical minterms and maxterms.
\end{remark}

\section{No complete linear term rewriting system for propositional logic}\label{sect:no-compl}

Recall that a linear inference is a sound linear rewrite rule. We
denote the set of all linear inferences by $\Lin$. 
We will now show that there is no sound linear term rewriting system
that is complete for $\Lin$ unless $\coNP = \NP$. 
The work in this section corresponds to point \eqref{item:reduction-to-nontriv-frag} in the Introduction, culminating in Theorem~\ref{thm:exists-polyder}, and ultimately point \eqref{item:np-alg-for-conp} by way of Corollary~\ref{cor:nolinsys}.

We start with the
following observation made in~\cite{Stra:08:Extensio:kk}:

\begin{proposition}\label{prop:coNP}
	$\Lin$ is $\coNP$-complete.
\end{proposition}

This result is the reason, from the point of proof theory, why one
might restrict attention to only linear inferences at all: every
Boolean tautology can be written as a linear inference.  As we can see
from the proof that follows, the translation is not very complicated,
and it induces an at most quadratic blowup in size from an input
tautology to a linear inference.

We include the proof here for completeness,
and also since the statement here differs slightly from that in
\cite{Stra:08:Extensio:kk}.

\begin{proof}[Proof of Proposition~\ref{prop:coNP}]  
	That $\Lin$ is in $\coNP$ is due to the fact that checking soundness of a
	rewrite rule $s\to t$ can be reduced to checking validity of the formula $\bar
	s\cor t$. To prove $\coNP$-hardness, we reduce validity of
	general tautologies to soundness of linear
	rewrite rules.  Let~$t'$ be the term obtained from~$t$ (which is assumed to be in NNF) by doing the following for each positive variable $x$: let $n$ be the
	number of occurrences of $x$ in $t$, and let $m$ be
	the number of occurrences of $\bar x$ in $t$. If $n=0$ replace every occurence of $\bar x$ by $\fff$, and if $m=0$ replace every occurrence of $x$ by $\fff$. Otherwise,
	introduce $2mn$ fresh (positive) variables
	$x'_{i,j},x''_{i,j}$ for $1\le i\le n$ and $1\le j\le m$. Now, for $1\le i\le n$, replace the $i$\textsuperscript{th} occurrence of $x$ by
	$x'_{i,1}\cor\ldots\cor x'_{i,m}$ and, for $1\le j\le
	m$, replace the $j$\textsuperscript{th} occurrence of $\bar x$ by $x''_{1,j}\cor\ldots\cor
	x''_{n,j}$. 

	Now $t'$ is a linear term (without negation),
	and its size is quadratic in the size of $t$. Let $s'$ be the
	conjunction of all pairs $x'\cor x''$ of variables introduced in the
	construction of $t'$. Clearly $\Var(s')=\Var(t')$ and $s'$ is also a
	linear term of the same size as $t'$. Furthermore, $t$ is a
	tautology if and only if $s'\to t'$ is sound. To see this, let
	$s''$ and $t''$ be obtained from $s'$ and $t'$, respectively, by
	replacing each $x''$ by $\bar x'$. Then $s''$ always evaluates to
	$1$, and $t''$ is a tautology if and only if $t$ is a tautology.
\end{proof}

In the next step we extend the result of the previous section to all
linear inferences, i.e., we have to deal with constants, negation,
erasure, and trivialities. Some of the following results appeared already in
\cite{Das:13:Rewritin:uq}, so we present only brief arguments here.

\begin{definition}\label{def:sm}
	We define the following rules:
	\[
	\swi : x \vlan (y \vlor z) \to (x\vlan y) \vlor z
	\qquad \quad
	\med : \vls[(w.x).(y.z)] \to \vls((w \vlor y).(x \vlor z))
	\]
	We call the former \emph{switch} and the latter \emph{medial} \cite{BrunTiu:01:A-Local-:mzF}.
\end{definition}

In what follows we implicitly assume that rewriting is conducted modulo $\ACU$.

\begin{lemma}\label{lem:triv-aside}
	If $s$ and $t$ are negation-free linear terms on a variable set $X$ of size $n$ and $s\leq t$, then there are linear terms $s' , t', u$ such that:
	\begin{enumerate}
		\item\label{item:shuffle-trivialities} There are derivations $s \longannotatedarrow{*}{ \swi ,\med } s' \vlor u$ and $t' \vlor u \longannotatedarrow{*}{ \swi ,\med } t$ of length $O(n^2)$. 
		\item $\fnimplies{s'}{t'}$ is sound and nontrivial.
	\end{enumerate}
\end{lemma}

\begin{proof}
	See \cite{Das:13:Rewritin:uq}. Briefly, the idea is that $u$ is obtained by repeatedly `moving aside' trivial variables, using $\swi ,\med $ and $\ACU$, until there are no trivialities remaining in~\hbox{$s'\to t'$}. The bound of $O(n^2)$ is not explicitly mentioned in~\cite{Das:13:Rewritin:uq}, but it is clear from direct inspection of that construction.
\end{proof}

\begin{remark}
Notice that, while the derivations from Lemma~\ref{lem:triv-aside}.\eqref{item:shuffle-trivialities} above are small in size, they are in general difficult to compute, due to the inherent complexity of detecting triviality. This problem is in fact already $\coNP$-complete, since validity of an arbitrary linear inference $s \to t$ can be reduced to detecting triviality at $x$ in $\vls(s.x)\to \vls[t.x]$, where $x$ is fresh. This is not an issue in what follows since we are only concerned with the existence of small derivations, and so the existence of an $\NP$-algorithm, for various inferences.
\end{remark}

A left- and right-linear rewrite rule may still erase or introduce variables, i.e.\ there may be variables on one side that do not occur on the other.\footnote{Usually, term rewrite rules are required to not introduce new variables from left to right, but it does no harm to make this generalisation here.}
However, notice that any such situation must constitute a triviality at such a variable, since the soundness of the step is not dependent on the value of that variable.

\begin{proposition}
\label{prop:nontrivial-implies-nonerasing-nonintroducing}
	Suppose $\rho: l\to r$ is linear, and there is some variable $x$ occurring in only one of\/ $l$ and $r$. Then $\rho$ is trivial at $x$.
\end{proposition}

If a (positive) variable $x$ occurs negatively on both sides of a linear rule then $\bar x$ can be replaced soundly by $x$ on both sides. Otherwise, if $x$ occurs positively on one side and negatively on the other, it must be that we have a triviality at $x$. 

\begin{proposition}\label{prop:neg-elim}
	For each linear rule $\rho$ either there is a negation-free linear rule that is equivalent to $\rho$ (i.e.\ with the same reduction steps), or $\rho$ is trivial.
\end{proposition}

Recall that $\ACU'$ preserves the Boolean function computed by a term, and that every linear term is $\ACU'$-equivalent to $\fff$, $\ttt$ or a unique constant-free linear term. 
Let us write $R\cdot S$ for the composition of relations $R$ and $S$, and $=_{\ACU'}$ for equivalence under $\ACU'$.

\begin{proposition}
\label{prop:nontrivial-implies-constant-free}
	If $R$ is a complete linear system then any constant-free nontrivial linear inference has a constant-free derivation in $=_{{\ACU'}} \cdot \annotatedarrow{}{R} \cdot =_{\ACU'}$.
\end{proposition}

\begin{proof}
	Let $s\to t$ be a constant-free nontrivial linear inference. By completeness there is an $R$-derivation of $\fnimplies{s}{t}$, in which we may simply reduce every line by $\ACU'$ to a constant-free term or $\fff$ or $\ttt$. However, if some line were to reduce to $\fff$ or $\ttt$ then either $s$ or $t$ would contain a constant, by soundness and Corollary~\ref{cor:acu'}, so the resulting sequence is a derivation of the appropriate format.
\end{proof}

Now, combining our results from Section~\ref{sect:inference} with the normal forms obtained above, we arrive at the main result of this work:

\begin{theorem}\label{thm:exists-polyder}
	If there is a sound and complete linear system for $\Lin$, 
	then there is one that has a $O(n^4)$-length derivation for each linear inference on $n$ variables.
\end{theorem}
\begin{proof}
  Assume we have a sound and complete linear system $R$ for $\Lin$, and let $s \to t$ be a linear inference on $n$ variables. 
  By Lemma~\ref{lem:triv-aside} we have linear terms $s',t'$ such that $|s'|\le|s|$ and $s'\to t'$ is sound, linear, and nontrivial. By Propositions~\ref{prop:nontrivial-implies-nonerasing-nonintroducing}, \ref{prop:neg-elim} and reduction under $\ACU'$ we can assume that $s', t'$ have the same size and are free of negation and constants.\footnote{If $s'$ or $t'$ is not equivalent to a constant-free term under $\ACU'$, then it is equivalent to $\bot$ or $\top$, whence we must have $s'=t'$ by non-triviality.}
  By Proposition~\ref{prop:nontrivial-implies-constant-free} there is thus a derivation of $s'\to t'$ in $=_{\ACU'} \cdot \annotatedarrow{}{R}\cdot =_{\ACU'}$ that is constant-free and negation-free. We can assume that each term in this derivation computes a distinct Boolean function, by Corollary~\ref{cor:acu'}, and so, by Theorem~\ref{thm:main-result}, the length of this derivation is $O(n^4)$. 
  Finally, by Lemma~\ref{lem:triv-aside}.\eqref{item:shuffle-trivialities}, this means that we can construct a derivation of $s\to t$ with overall length $O(n^4)$ in $R\cup \set{ \swi ,\med }\cup \ACU'$.
\end{proof}

\begin{corollary}\label{cor:nolinsys}
	There is no sound linear system complete for $\Lin$ unless $\coNP = \NP$.
\end{corollary}
\begin{proof}
  By Proposition~\ref{prop:coNP}, $\Lin$ is $\coNP$-complete, and the existence of such a system would lead to a $\NP$ decision procedure for $\Lin$ by
  Theorem~\ref{thm:exists-polyder}: for any linear inference on $n$ variables we could simply guess a correct $O(n^4)$ length derivation in an appropriate system.
\end{proof}

\section{On the canonicity of switch and medial}\label{sect:canon}

In this section we investigate to what extent the two rules switch
and medial from Definition~\ref{def:sm}, which play a crucial role in the proof theory of classical
propositional logic, are ``canonical''. 
Let us restrict our attention to constant-free terms and rules for this section.

Recall that the switch and medial rules are as follows:
	\[
	\swi : x \vlan (y \vlor z) \to (x\vlan y) \vlor z
	\qquad \quad
	\med : \vls[(w.x).(y.z)] \to \vls((w \vlor y).(x \vlor z))
	\]
First we observe that both rules are minimal in the following sense:

\begin{definition}
  A sound linear rewrite rule $\rho\colon l\to r$ is \emph{minimal} if
  there is no linear term $t$ on the same variables as $l$ and $r$
  such that $l<t<r$.
\end{definition}

\begin{proposition}
Switch and medial are minimal.
\end{proposition}

\begin{proof}
  By exhaustive search on all terms of size 3 (for switch) and 4 (for medial).
\end{proof}

Observe that, seen as an action on relation webs, switch and medial preserve $\vee$-edges and $\wedge$-edges, respectively. Formally, let us consider the following two properties of a linear inference $\rho$:
\begin{itemize}[label=(**)]
\item[(*)]\label{item:switch-preserves-dis-edges} If $s \annotatedarrow{}{\rho} t$ then, whenever $\gedge[]xy$ in $\web s$, we have that $\gedge[]xy$ in $\web t$.
\item[(**)]\label{item:medial-preserves-con-edges} If $s \annotatedarrow{}{\rho} t$ then, whenever $\redge[]xy$ in $\web s$, we have that $\redge[]xy$ in $\web t$.
\end{itemize}

Our first canonicity result is that medial is the \emph{only}
sound linear inference that is minimal and satisfies (**). 
In fact, we will show the stronger property that any sound linear rule satisfying (**) is already derivable by medial.
First, we will require a certain relation between the webs of terms, which was defined in \cite{Stra:07:A-Charac:fk}.

\begin{definition}
	Let $s$ and $t$ be linear terms on a set $X$ of variables. We write $\MedialCriterion{s}{t}$\/ if:
	\begin{enumerate}
		\item Whenever $\redge[]xy$ in $\web s$ we have that $\redge[]xy$ in $\web t$.
		\item Whenever $\gedge[]xy$ in $\web s$ and $\redge[]xy$ in $\web t$, there are $w,z\in X$ such that,
		\begin{equation*}
		\FourGraph{w,x,y,z}rggggr \text{ in $\web s$}
		\quad \text{and} \quad
		\FourGraph{w,x,y,z}rgrrgr \text{ in $\web t$.}
		\end{equation*}
	\end{enumerate}
\end{definition}

This relation allows us to relate structural properties of graphs to derivability by medial, via the characterisation result below. The proof from \cite{Stra:07:A-Charac:fk} relies on careful analysis of subterms which is  beyond the scope of this paper.

\begin{proposition}[Medial criterion]\label{MedialCriterion}
	$\MedialCriterion{s}{t}$ if and only if $s\annotatedarrow{*}{\med}t$. 
\end{proposition}

Using this result we can show that any sound linear rule satisfying (**) is already derivable by medial:

\begin{theorem}\label{thm:medial}
	Let $s$ and $t$\/ be linear terms on a variable set $X$. The following are equivalent:
	\begin{enumerate}
		\item\label{MedialThm-RedPres+Sound} $s\leq t$ and  for all $x,y\in X$ we have $\redge[]xy$ in $\web s$ implies $\redge[]xy$ in $\web t$.
		\item\label{MedialThm-RedPres+Crit} $\MedialCriterion{s}{t}$.
		\item\label{MedialThm-Medial} $s\annotatedarrow{*}{\med} t$.
	\end{enumerate}
\end{theorem}

For the proof let us say, if $t$ is a linear term with $x,y,z\in\Var(t)$, that \emph{$y$ separates $x$ from $z$} in $\web{t}$ if $\TwoGraph{x,y}{r}$ in $\web t$ and $\TwoGraph{y,z}{g}$ in $\web t$.

\begin{proof}
[Proof of Theorem~\ref{thm:medial}]
	We have that \ref{MedialThm-RedPres+Crit}$\implies$\ref{MedialThm-Medial} by Proposition~\ref{MedialCriterion} and \ref{MedialThm-Medial}$\implies$\ref{MedialThm-RedPres+Sound} by inspection of medial, so it suffices to show \ref{MedialThm-RedPres+Sound}$\implies$\ref{MedialThm-RedPres+Crit}.
	For this, assume \ref{MedialThm-RedPres+Sound} and 
	suppose $\TwoGraph{x,y}{g}$ in $\web s$ and $\TwoGraph{x,y}{r}$ in $\web t$, and let $S$ be a minterm of $s$ containing $x$. We must have $S \supsetneq \{x\}$ since $\TwoGraph{x,y}{r}$ in $\web{t}$ and $s\to t$ is sound.\footnote{By Proposition~\ref{prop:sound-minmax} and Theorem~\ref{thm:rel-minmax-webs}, there must a subset of $S$ which is a maximal $\wedge$-clique in $\web{t}$.} Similarly there must be a maxterm $T$ of $t$ containing $y$ such that $T\supsetneq \{y\}$. Now, by \ref{MedialThm-RedPres+Sound}, it must be that $S$ (resp.\ $T$) is also a minterm (resp.\ maxterm) of $t$ (resp.\ $s$),\footnote{Since by~\ref{MedialThm-RedPres+Sound}, $\wedge$-edges (resp.\ $\vee$-edges) are preserved left-to-right (resp.\ right-to-left) and so $\wedge$-cliques (resp. $\vee$-cliques) must be preserved (resp.\ reflected). Of course, these must be maximal by soundness.} and so, by Theorem~\ref{thm:gurevich}, there is some (unique) $z\in S\cap T$ which, by definition, separates $x$ from $y$ in both $\web s$ and $\web t$. By a symmetric argument we obtain a $w$ separating $y$ from $x$ in both $\web s$ and $ \web t$. 
	By construction, $w$ and $z$ must be distinct, so we have the following situation,
	\[
	\FourGraph{x,z,w,y}rggugr \text{ in $\web{s}$}
	\quad\text{and}\quad
	\FourGraph{x,z,w,y}rgrugr \text{ in $\web{t}$.}
	\]
	whence \ref{MedialThm-RedPres+Crit} follows by $P_4$-freeness.
\end{proof}

\begin{corollary}[Canonicity of medial]\label{cor:medial}
  Medial is the only sound
linear inference that is minimal and has property (**).
\end{corollary}

\begin{proof}
  By Theorem~\ref{thm:medial}, any linear inference satisfying (**) can be derived by medial. The result then follows by minimality of medial.
\end{proof}

Using these results, we are actually able to improve the length bound on nontrivial linear derivations that we proved earlier:

\begin{corollary}\label{cor:cubic-bound}
	The bound in Theorem~\ref{thm:main-result} can be improved to $O(n^3)$.
\end{corollary}

For the proof, let us first define $\numand (t)$ (resp.\ $\numor (t)$) to be the number of $\wedge$ (resp.\ $\vee$) symbols occurring in $t$.

\begin{proof}[Proof of Corollary~\ref{cor:cubic-bound}]
	Instead of using $\redmeas$ in Lemma~\ref{lem:increasingmeasure}, use $\numor$, which is linear in the size of the term. If no $\wedge$-edge changes to a $\vee$-edge in some step, it follows by Theorem~\ref{thm:medial} that the step is derivable using medial, and so $\numor$ must have strictly increased.
\end{proof}

While we have just shown a fairly succinct form of canonicity for medial, it turns out that we cannot obtain an analogous result for switch: switch is \emph{not} the only sound linear
inference that is minimal and satisfies (*). To see this, simply recall the example of \eqref{eqn:php32-linear} from the Introduction:
\begin{equation*}
\begin{array}{rl}
& ( u \vlor (v \vlan v') )
\vlan 
( ( w \vlan w' ) \vlor (x \vlan x') )
\vlan
((y \vlan y') \vlor z ) 
\\
\noalign{\smallskip}
\to & (u \vlan ( w \vlor y )) \vlor ( w' \vlan y' )
\vlor
(v' \vlan x') \vlor ( (v\vlor x) \vlan z )
\end{array}
\end{equation*}
Notice, however, that this inference does not preserve the number~$\numand$ of conjunction symbols in a term.
In fact, switch is the only nontrivial linear inference we know of that preserves~$\numand$, although there are known trivial examples that even \emph{increase} $\numand$, for instance the ``supermix'' rules from \cite{Das:13:Rewritin:uq} that we considered earlier in Example~\ref{ex:triviality}, \eqref{eqn:smix}:
\[
x \vlan ( y_1 \vlor \vldots \vlor y_n )
\quad \to \quad
x \vlor (y_1 \vlan \vldots \vlan y_n)
\]
This leads us to the following conjecture: 
\begin{conjecture}\label{conj:switch-canon}
	If $s\to t$ is sound, nontrivial, satisfies (*) and $\numand (s) \leq \numand (t)$, then $s\annotatedarrow{*}{\swi} t$. 
\end{conjecture}

Notice that this conjecture would already imply our main result, Theorem~\ref{thm:main-result}, since $\numand \times \redmeas$ would be a strictly decreasing measure.
This measure can also be used for the usual proof of termination of $\{\swi ,\med \}$ (constant-free and modulo $\AC$) and also yields a cubic bound on termination.\footnote{In fact, using a different measure, it can also be shown that $\{\swi ,\med \}$ terminates with a quadratic bound.} We point out that, in this work, we have matched that bound for \emph{all} linear derivations that are not trivial.

The supermix rules are also examples of linear inferences that satisfy neither (*) nor (**). However, again, we have not been able to identify any nontrivial examples of this, and we further conjecture the following:
\begin{conjecture}\label{con:minimal}
  There is no nontrivial minimal sound linear inference that satisfies neither \emph{(*)} nor \emph{(**)}. 
\end{conjecture}

An interesting observation is that Conjecture~\ref{con:minimal} and
Corollary~\ref{cor:medial} together entail that medial is the only
linear inference that allows contraction to be reduced to atomic form. To see what this means, consider again~\eqref{eqn:conred} from
the introduction. The steps marked $\gcd$ are instances of the contraction rule $x\vlor x\to x$. 
If the contractum of such a step is simply a variable, then we call that instance of contraction \emph{atomic}, denoted by
$\acd$ as in \cite{BrunTiu:01:A-Local-:mzF}. 
Dually, the atomic instances of `cocontraction' $x\to \vls(x.x)$, when the redex is simply a variable,
are denoted by $\acu$. 
We say that a linear inference $\rho : l \to
r$ \emph{reduces contraction to atomic form} if, for every
term $t$, we have $t \vlor t \longannotatedarrow{*}{\rho, \acd } t$ and $t
\longannotatedarrow{*}{\rho, \acu } t \vlan t$, modulo $\ACU$.

\begin{conjecture}
  Medial is the only minimal linear inference that reduces contraction to atomic form. More precisely, for every linear
  inference $\rho : l \to r$ that reduces contraction to
  atomic form we have $l \annotatedarrow{*}{\med } r$.
\end{conjecture}

\begin{proof}
[Proof using Conjecture~\ref{con:minimal}]
  Assume $t \vlor t \longannotatedarrow{*}{\rho, \acd } t$ and $t
  \longannotatedarrow{*}{\rho, \acu } t \vlan t$ modulo $\ACU$, for every term $t$. Since $t$ can contain
  $\vee$ and $\wedge$, it must be the case that $\rho$ replaces
  $\vee$-edges in $\web{l}$ by $\wedge$-edges in $\web{r}$. By
  Conjecture~\ref{con:minimal} $\rho$ does not replace $\wedge$-edges in
  $\web{l}$ by $\vee$-edges in $\web{r}$. By Theorem~\ref{thm:medial}
  we must have $l \annotatedarrow{*}{\med }
  r$.
\end{proof}

\newcommand{\one}{}
\newcommand{\two}{}
\newcommand{\mk}[1]{{#1}^{\scriptscriptstyle\bullet}}
\newcommand{\rwdcd}{{{\mathsf w}{\downarrow}{\hbox{-}}{\mathsf c}{\downarrow}}}
\newcommand{\rwdiu}{{{\mathsf w}{\downarrow}{\hbox{-}}{\mathsf i}{\uparrow  }}}
\newcommand{\rwdwu}{{{\mathsf w}{\downarrow}{\hbox{-}}{\mathsf w}{\uparrow  }}}
\newcommand{\rwdcu}{{{\mathsf w}{\downarrow}{\hbox{-}}{\mathsf c}{\uparrow  }}}
\newcommand{\rcuwu}{{{\mathsf c}{\uparrow  }{\hbox{-}}{\mathsf w}{\uparrow  }}}
\newcommand{\rcdwu}{{{\mathsf c}{\downarrow}{\hbox{-}}{\mathsf w}{\uparrow  }}}
\newcommand{\rcdiu}{{{\mathsf c}{\downarrow}{\hbox{-}}{\mathsf i}{\uparrow  }}}
\newcommand{\rcdcu}{{{\mathsf c}{\downarrow}{\hbox{-}}{\mathsf c}{\uparrow  }}}
\newcommand{\ridwu}{{{\mathsf i}{\downarrow}{\hbox{-}}{\mathsf w}{\uparrow  }}}
\newcommand{\ridcu}{{{\mathsf i}{\downarrow}{\hbox{-}}{\mathsf c}{\uparrow  }}}
\newcommand{\rcucd}{{{\mathsf c}{\uparrow}{\hbox{-}}{\mathsf c}{\downarrow }}}

\section{On the normalisation of deep inference proofs}
\label{sect:norm-di}

Another application of our results is to the normalisation of deep inference proofs. This is typically done via rewriting on certain graphs extracted from derivations, known as \emph{atomic flows}~\cite{GuglGund:07:Normalis:lr,gug:gun:str:LICS10}. The main sources of complexity here are `contraction loops', and so a lot of effort has gone into the question of whether such features can be eliminated. A consequence of our main result is that this is 
impossible for a large class of deep inference systems.

\newcommand{\Mon}{\text{\textsc{Mon}}}
\newcommand{\MSKS}{\mathsf{MSKS}}
\newcommand{\GNorm}{\mathsf{norm}}
\newcommand{\flow}{\mathit{fl}}

We will now only consider rewriting systems on positive terms, and then make some remarks about negative rules at the end of this section. We consider systems with the standard structural rules of deep inference, extended by an arbitrary (polynomial-time decidable) set of linear rules.

A formal definition of atomic flows can be found in \cite{GuglGund:07:Normalis:lr}, where they were first presented, and an alternative presentation can be found in~\cite{gug:gun:str:LICS10}. We give an informal definition below which is sufficient for our purposes.

\newcommand{\vcengraph}[1]{\raisebox{-0.5\height}{#1}}

\begin{definition}
[Structural rules and atomic flows]
We define the system $\gc  \gw$ as follows:
\[
\begin{array}{rclrcl}
\gwd & :&   x \to  x \vlor y
\quad & \quad
\gwu &: & x \vlan y \to x
\\
\gcu & : & x \to x \vlan x
\quad & \quad 
\gcdown & : & x \vlor x \to x
\end{array}
\]
If $S$ is the extension of $\gc  \gw  $ by a set of linear rules and $\pi$ is an $S$-derivation (written as a vertical list), then the \emph{atomic flow} of $\pi$, denoted $\flow (\pi)$, is the (downwards directed) graph obtained by tracing the paths of each variable through the derivation, designating nodes at $\gc  \gw  $ steps as follows:
\[
\begin{array}{rccrcc}
\gwd & :&  \vcengraph{\includegraphics{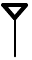}}
\quad & \quad
\gwu &: & \vcengraph{\includegraphics{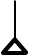}}
\\
\noalign{\bigskip}
\gcu & : & \vcengraph{\includegraphics{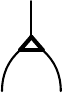}}
\quad & \quad 
\gcdown & : & \vcengraph{\includegraphics{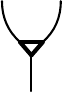}}
\end{array}
\]

\end{definition}

\begin{example}
		Consider the system $\MSKS$ obtained by extending $\gc  \gw  $ by the rules switch and medial, from Definition~\ref{def:sm}, as well as rules $\ACU$ from Section~\ref{sect:prelim-rew} for associativity, commutativity and constants. This is equivalent to the monotone fragment of the common deep inference system $\SKS$~\cite{BrunTiu:01:A-Local-:mzF}.

Here is an example of an $\MSKS$ rewrite derivation, with redexes underlined, and its atomic flow. The colours are used to help the reader associate edges with variable occurrences in the derivation. 
\begin{equation}
\label{eqn:example-der+flow}
\vcengraph{\includegraphics{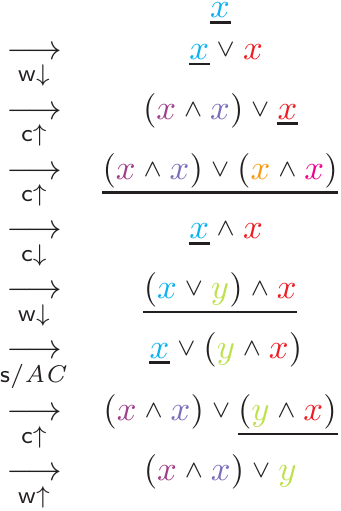}}
\qquad\qquad
\vcengraph{\includegraphics{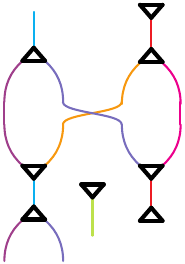}}
\end{equation}
\end{example}

\begin{definition}
	[Flow rewriting systems]
	A \emph{flow rewriting system} (FRS) is a set of graph rewriting rules on atomic flows. We say that a FRS $R$ \emph{lifts} to a TRS $S$ if, for every $S$-derivation $\pi : s \annotatedarrow{*}{S} t$ and reduction step $\flow (\pi) \to \phi$ there is a $S$-derivation $\pi' : s \annotatedarrow{*}{S} t$ with $\flow (\pi') = \phi$.
\end{definition}

\begin{example}
 Consider the following FRS, which is a subset of rules occurring in \cite{GuglGund:07:Normalis:lr,gug:gun:str:LICS10} and which is called $\GNorm$ in \cite{Das:13:The-Pige:fk}.
	\begin{equation}
	\label{eqn:frs-norm}
\vcengraph{	\includegraphics[scale=0.9]{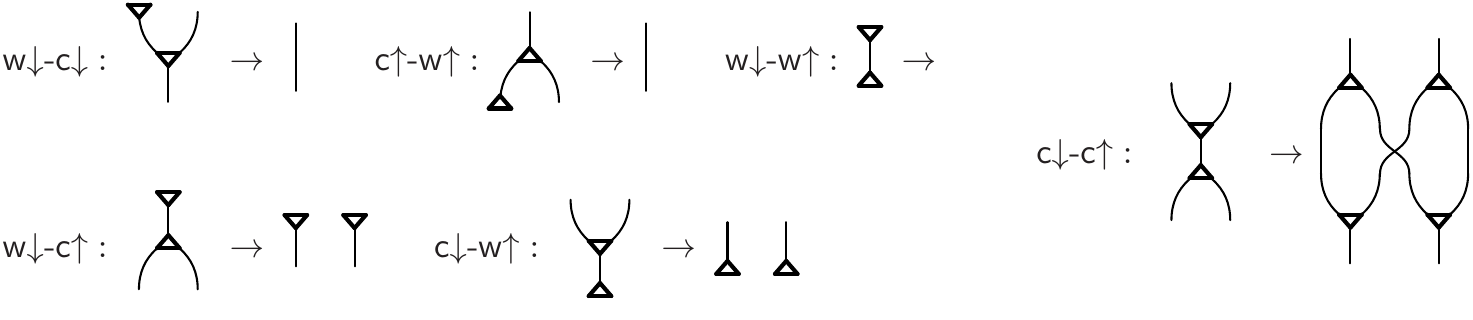}}
	\end{equation}
			We have essentially the following result from \cite{GuglGund:07:Normalis:lr}:
		\begin{proposition}
		\label{prop:norm-lifts-to-msks-extensions}
		$\GNorm$ lifts to any extension of $\MSKS$ by linear rules.
		\end{proposition}
		The proof of this is beyond the scope of this work, but crucially relies on the presence of switch, medial and $\ACU$ to make the $\gw  $ and $\gc  $ rules atomic, cf.~\ref{eqn:conred}, and thereby allow these steps to permute more freely in a derivation.
		
	 For example, here is a $\GNorm$-derivation that normalises the flow from \eqref{eqn:example-der+flow},
	 \begin{equation}
	 \label{eqn:example-flow-rew-der}
\vcengraph{	 \includegraphics{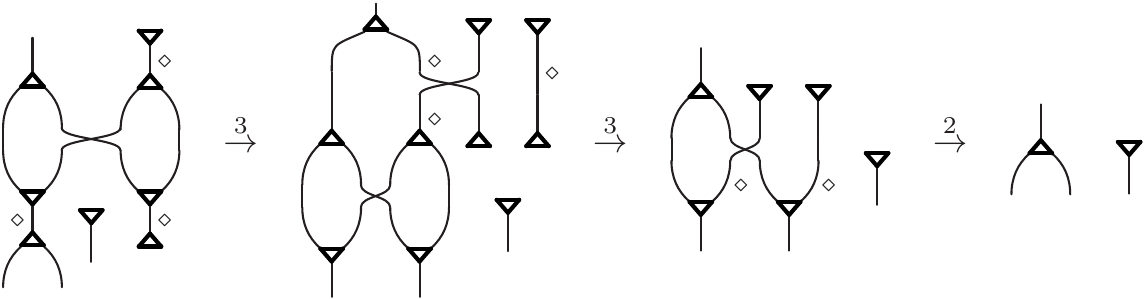}}
	 \end{equation}
	 where redexes are marked by $\diamond$.
\end{example}

	$\GNorm$ is strongly normalising, as implied by results in \cite{GuglGund:07:Normalis:lr}. In the works \cite{Das:12:Complexi:kx} and \cite{Das:13:Some-Res:fk} the main source of complexity of (weak) normalisation under $\GNorm$ is the presence of \emph{contraction loops}. In their absence the time complexity of normalisation is polynomially bounded.
	
	\begin{definition}
		[Contraction loops, from \cite{Das:12:Complexi:kx}]
		Given a flow $\phi$, a \emph{contraction loop} is a pair of nodes $(\nu_1 , \nu_2) $ such that there are two distinct paths from $\nu_1$ to $\nu_2$ in $\phi$.
	\end{definition}
	
	It turns out that our previous results imply that no deep inference system that extends $\MSKS$ by linear rules can admit a flow-rewriting normalisation procedure that eliminates contraction loops:
	
	\begin{theorem}
	\label{thm:no-sound-frs-elim-loops}
		Let $R$ be a FRS such that, for any flow $\phi$, there is some flow $\psi$ free of contraction loops such that $\phi \annotatedarrow{*}{R} \psi$. Then $R$ lifts to no sound system extending $\MSKS$ by linear rules unless $\coNP = \NP$.
	\end{theorem}

	Before giving the proof, let us first make the following observation:
	\begin{proposition}
	\label{prop:norm-preserves-loop-freeness}
	If a flow $\phi$ is free of contraction loops and $\phi \longannotatedarrow{*}{\GNorm} \psi$, then $\psi$ is also free of contraction loops.
	\end{proposition}
	\begin{proof}
	[Proof sketch]
	By induction on the length of a $\GNorm$-derivation under a careful analysis of the reduction steps in $\GNorm$.
	\end{proof}
	
	We can now give a proof of the theorem above.

	\begin{proof}
	[Proof of Theorem~\ref{thm:no-sound-frs-elim-loops}]
          Let us assume that $R$ lifts to such a system $S$ and show that $\coNP = \NP$. Let $s \to t$ be an arbitrary linear inference and let $s', t' , u$ be linear terms obtained by Lemma~\ref{lem:triv-aside}. By completeness of $S$ let $\pi : s' \annotatedarrow{*}{S} t'$ and let $\pi': s' \annotatedarrow{*}{S}t'$ be obtained by first reducing $\flow(\pi)$ under $R$ to a flow free of contraction-loops and then to a normal form under $\GNorm$,
          and finally lifting the resulting derivations to $S$ by assumption and Proposition~\ref{prop:norm-lifts-to-msks-extensions}. Notice that $\flow(\pi')$ is free of contraction loops by assumption and Proposition~\ref{prop:norm-preserves-loop-freeness}.
		
		First we show that $\flow (\pi')$ must be free of $\gcdown$ and $\gcu $ nodes.
		 Consider a topmost $\gcdown $ node and the maximal paths leading to its upper edges. Since $\flow (\pi') $ is free of contraction loops we can assume these two paths are disjoint. If one of the paths begins with a $\gwd $ node then there must be either a $\rwdcd$ or $\rwdcu$ redex in $\flow(\pi')$, contradicting normality under $\GNorm$. Therefore both paths must begin with variables from $s'$, contradicting linearity of $s'$. The argument for $\gcu $ is similar, by consideration of a bottommost such node.
		
		Now we show that $\flow (\pi')$ is free of $\gwd $ and $\gwu $ nodes. 
		Suppose there is a $\gwu $ node and consider the maximal path leading to its edge.
		 This cannot be connected to any other node since this would yield a redex. Therefore this path must begin from some variable $x$ of $s'$. Consequently the occurrence of $x$ in $t'$ must originate from a $\gwd $ node.\footnote{Recall that we already have that there are no $\gcdown $ or $\gcu $ nodes, so this follows immediately.} However this would imply that $s'\to t'$ is trivial at $x$, contradicting the fact that $s' \to t'$ is nontrivial.
		
		Therefore $\flow(\pi')$ is just a flow of simple edges, and so $\pi'$ is linear. Since it also derives a nontrivial linear inference, it must have polynomial length by Theorem~\ref{thm:main-result}. Finally, by Lemma~\ref{lem:triv-aside}, this means that there is a polynomial-size $S$-derivation of $s \to t$. Since the choice of this linear inference was arbitrary, we thus have an $\NP$ algorithm for $\Lin$.
	\end{proof}

In particular we can conclude that a particularly natural FRS for eliminating contraction loops cannot be correct for a large class of deep inference systems, partially answering questions occurring in previous works and correspondences:

\begin{corollary}
The following flow-rewriting rule,
	\[
\rcucd \quad : \quad 
\vcengraph{\includegraphics{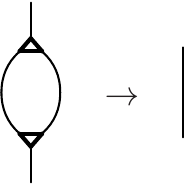}}
	\]
lifts to no sound system extending $\MSKS$ by linear rules unless $\coNP = \NP$.
\end{corollary}

\renewcommand{\assoc}{\mathsf{assoc}}

The proof follows immediately from Theorem~\ref{thm:no-sound-frs-elim-loops} and the following observations:
\begin{proposition}
We have the following:
\begin{enumerate}
\item\label{item:con-assoc-lifts} The equivalence relation $\assoc$ generated from the following equations,
\[
\includegraphics{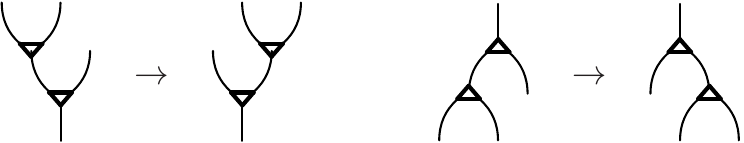}
\]
lifts to any extension of $\MSKS$ by linear rules.
\item\label{item:rcucd+assoc-elim-loops} Any flow can be reduced in $\rcucd + \rcdcu+ \assoc$ to one free of contraction loops. %
\end{enumerate}
\end{proposition}
\begin{proof}
[Proof sketch]
\ref{item:con-assoc-lifts} is routine, so we prove \ref{item:rcucd+assoc-elim-loops}. For a $\gcu $ node in a flow, let its \emph{weight} be its distance from the top of the flow. We argue that $\rcdcu + \rcucd$ is terminating modulo $\assoc$, by noticing that the multiset of weights of $\gcu $ nodes in a flow decreases\footnote{Formally it suffices to associate a flow $\phi$ with the sum $\sum 2^{2w(\nu)}$, where $\nu$ ranges over $\gcu $ nodes in $\phi$ and $w(\nu)$ is the weight of $\nu$, and consider the usual order on natural numbers.} by any application of $\rcdcu$ or $\rcucd$ and is preserved by $\assoc$. Finally, we observe that there cannot be any contraction loop in a normal form of $\rcdcu + \rcucd $ modulo $\assoc$ since it would contain either a $\rcdcu$ or $\rcucd$ redex, modulo $\assoc$.
\end{proof}

\begin{remark}
Here we only considered systems that extend the monotone fragment of the deep inference system $\SKS$ by arbitrary linear rules. To some extent the results above generalise to extensions by other rules, but there are certain interesting cases that could be points of further study.

First, of course, there could be rules that allow an interplay between positive and negative variables, most notably the identity and cut rules from $\SKS$:
\[
\ttt \quad \to \quad x \vlor \bar x
\qquad
\qquad 
x \vlan \bar x \quad \to \quad \fff
\]
Their normalisation behaviour is very different from that of the structural rules contraction and weakening, and so call for an independent analysis altogether. \footnote{We are aware that work studying linear systems extended by such rules is currently being pursued by Guglielmi, McCusker and Santamaria. This line of research is also related to \cite{lamarche:gap} and~\cite{str:medial}.}

Another interesting case is when $\SKS$ is extended by nonlinear rules. In a particularly extreme case one can envisage rules that are `multiplicative' but not linear. For instance, consider the following monotone formula, denoted $t(w,x,y,z)$:
\[
(w \vlan x)
\vlor 
((w \vlor x) \vlan (y \vlor z))
\vlor
(y \vlan z)
\]
This computes the threshold function $\mathit{TH}^X_2 $ from Example~\ref{ex:bool-fn-summary}, for $X = \{w,x,y,z\}$. Since this is a symmetric function, we can construct the following sound rule:\footnote{In fact it would be sound for any permutation of variables, but this is the prototypical interesting case.}
\[
t(w,x,y,z)
\quad \to \quad
t(w,y,x,z)
\]
It can be considered `multiplicative', in the sense that each variable occurs with the same multiplicity, $2$, on each side, but it cannot be an instance of a linear rule, since we rely on the logical dependencies between variable occurrences for soundness. 
\end{remark}

\newcommand{\conentails}{\annotatedarrow{}{\vlan}}
\newcommand{\disentails}{\annotatedarrow{}{\vlor}}
\section{Towards proof theory on arbitrary graphs}
\label{sect:graph-logic}

In this section we consider arbitrary complete undirected graphs with edges labelled by $\wedge $ and $\vee$, i.e.\ graphs that are not necessarily $\pfour$-free, and we consider their $\wedge$-maxcliques and $\vee$-maxcliques.
Such graphs no longer correspond to terms, in fact they do not even correspond to Boolean functions since Theorem~\ref{thm:rel-minmax-webs} breaks down by the example of~\eqref{eq:P4}:
\[
	\FourGraph{w,x,y,z}grrgrg
\]
The problem here is that there is a $\wedge$-maxclique $\{w,z\}$ and a $\vee$-maxclique $\{x,y\}$ which are disjoint, so under the association of $\wedge$- and $\vee$-maxcliques to minterms and maxterms respectively via Theorem~\ref{thm:rel-minmax-webs}, one would be able to force this graph to evaluate to $0$ and $1$ simultaneously by the assignment $\{ w \mapsto 1 , x \mapsto 0 , y\mapsto 0, z\mapsto 1 \}$.

On the other hand, the alternative definitions of entailment from Proposition~\ref{prop:sound-minmax} still remain meaningful in such a setting. Inspired by this, let us consider the following relations on graphs:
\begin{itemize}
\item $G \conentails G'$ if, for any $\wedge$-maxclique $C$ of $G$, there is a $\wedge$-maxclique $C'$ of $G'$ with $C' \subseteq C$.
\item $G \disentails G'$ if, for any $\vee$-maxclique $C$ of $G'$, there is a $\vee$-maxclique $C'$ of $G$ with $C' \subseteq C$.
\end{itemize}
They have the following important properties, whose proofs are routine:
\begin{proposition}
	\label{prop:graph-logic-ref-trans}
$\conentails$ and $\disentails$ are reflexive and transitive.
\end{proposition}

The point here is that, even though maximal cliques no longer correspond to minterms and maxterms, the notion of entailment induced by maximal cliques remains stable: if one starts with a $\pfour$-free graph and applies one of the relations $\conentails$ or $\disentails$ iteratively, and finishes with a $\pfour$-free subgraph, then the underlying implication is sound, even if many of the intermediate graphs are not $\pfour$-free, and so do not correspond to Boolean functions at all.

\newcommand{\FiveGraphCircCont}[2]{
	\raisebox{-0.45\height}{\begin{tikzpicture}
		\node[vertex] (v1) at (0,1) {$\StrBetween[1,2]{,\fir,}{,}{,}$};
		\node[vertex] (v2) at (0.951,0.309) {$\StrBetween[2,3]{,\fir,}{,}{,}$};
		\node[vertex] (v3) at (0.588,-0.809) {$\StrBetween[3,4]{,\fir,}{,}{,}$};
		\node[vertex] (v4) at (-0.588,-0.809) {$\StrBetween[4,5]{,\fir,}{,}{,}$};
		\node[vertex] (v5) at (-0.951,0.309) {$\StrBetween[5,6]{,\fir,}{,}{,}$};
		\draw[\sec] (v1) -- (v2) ;
		\draw[\thi] (v1) -- (v3);
		\draw[\fou] (v1) -- (v4);
		\draw[\fif] (v1) -- (v5);
		\draw[\six] (v2) -- (v3);
		\draw[\sev] (v2) -- (v4);
		\draw[\eig] (v2) -- (v5) ;
		\draw[\nin] (v3) -- (v4) ;
		\draw[#1] (v3) -- (v5) ;
		\draw[#2] (v4) -- (v5) ;
		\end{tikzpicture} }}
\newcommand{\FiveGraphCirc}[9]{
	\def\fir{#1}
	\def\sec{#2}
	\def\thi{#3}
	\def\fou{#4}
	\def\fif{#5}
	\def\six{#6}
	\def\sev{#7}
	\def\eig{#8}
	\def\nin{#9}
	\FiveGraphCircCont
}

For instance, consider the following reduction:
\[
\FiveGraphCirc{v,w,x,y,z}rgggrggrgr
\quad \to \quad
\FiveGraphCirc{v,w,x,y,z}rggrrggrgr
\]
This can easily be seen to be an instance of $\conentails$, since only a new $\wedge$-maxclique, $\{v,z\}$, is added. On the other hand, its \emph{inverse} is an instance of $\disentails$. Consequently the relations $\conentails$ and $\disentails$ really are distinct, unlike their restrictions to $\pfour$-free graphs.

\begin{remark}
Notice that there are alternative ways to define entailment for Boolean terms via their webs, but other intuitive choices do not satisfy Proposition~\ref{prop:graph-logic-ref-trans} when generalised to arbitrary graphs in the natural way, and so do not induce any meaningful logic. For example, for linear terms $s$ and $t$, we can show that $s \leq t$ if and only if every $\wedge$-maxclique of $\web{s}$ intersects every $\vee$-maxclique of $\web{t}$.\footnote{If $s$ evaluates to $1$, then one of its minterms must entirely be assigned to $1$, and if this intersects every maxterm of $t$, then no maxterm of $t$ is entirely assigned to $0$, so $t$ must also evaluate to $1$. Conversely, if some minterm of $s$ and some maxterm of $t$ do not intersect, then we can simultaneously force $s$ to evaluate to $1$ and $t$ to evaluate to $0$.} However, when generalised to arbitrary graphs, this relation is not even reflexive because of, again, the case of a $\pfour$ configuration~\eqref{eq:P4}.
\end{remark}

In further work we would like to study the logics induced by the relations $\conentails$ and $\disentails$, and even systems where one may alternate between them any time a graph is, say, $\pfour$-free. Such systems would be sound for Boolean logic when the source and target are $\pfour$-free, under the association of a term to its web. %
They would also leave the world of Boolean functions altogether, as we previously mentioned, which bears semblance to algebraic proof systems for propositional logic such as Cutting Planes and Nullstellensatz (studied in, for example, \cite{BonetPR97:cutting-planes} and \cite{BussIPRS97:nullstelensatz}).

Furthermore, notice that our crucial Lemma~\ref{intersectionlemma} cannot immediately be generalised to the setting of arbitrary graphs due to the fact that $\wedge$-maxcliques no longer necessarily intersect $\vee$-maxcliques. It would be particularly interesting to examine the extent to which `linear reasoning' can be recovered in this setting, sidestepping the shortcomings of $\pfour$-free graphs (i.e.\ terms) we have studied in this work.

\section{Final remarks}\label{sect:conc}
To some extent, this work can be seen as a justification for the approach of `structural' proof theory: for any deductive system that can be embedded into a rewriting framework on Boolean terms, as we have considered here, completeness requires the inclusion of structural rules that introduce, destroy and duplicate formulae, unless $\coNP= \NP$. It is not difficult to see that this covers a large class of proof systems, including essentially all the well-known systems based on formulae or related structures, e.g.\ Gentzen sequent calculi, Hilbert-Frege systems, Resolution, deep inference systems etc. On the other hand, as we mentioned in Section~\ref{sect:graph-logic}, proof systems based on other objects such as algebraic equations or graphs are not covered by our result.
While the observation that structural behaviour is somewhat necessary for proof theory is perhaps not surprising, it is of natural theoretical interest.

There are clear thematic relationships between this line of work and linear logic. In some ways, we can see this work as contributing to the study of the `multiplicative' fragment of Boolean logic. One particular connection we would like to point out is with Blass' model of linear logic in \cite{Blass92:game-semantics-ll}, the first game semantics model of linear logic. The multiplicative fragment of this model in fact validates precisely the sound linear inferences of Boolean logic\footnote{Under the assiociation of $\vlte$ with $\wedge$ and $\vlpa$ with $\vee$.}, which he calls `binary tautologies'. 
Following from the paragraph above, it would seem that one drawback of this model is that it can admit no sound and complete proof system, unless $\coNP = \NP$, by virtue of our results.

Finally, this work contributes to the study of term rewriting systems for Boolean Algebras. While complete axiomatisations have been known since the early 20th century by Whitehead, Huntington, Tarski and others, these are typically sets of equations, rather than `directed' rewrite rules which are more related to proof theory. It has been known for some time, for example, that there is no convergent TRS for Boolean Algebras \cite{Socher-Ambrosius91:no-conv-trs-ba}; our result, in the same vein, shows there is no \emph{linear} TRS for the linear fragment of Boolean Algebras.

\bibliographystyle{alpha}
\bibliography{biblio}

\end{document}